\documentclass[12pt,a4paper]{article}%
\usepackage{amssymb}
\usepackage{amsfonts}
\usepackage{amsmath}
\usepackage{graphicx}
\usepackage[singlespacing]{setspace}
\usepackage{geometry}
\usepackage{url}
\usepackage{hyperref}%
\providecommand{\U}[1]{\protect\rule{.1in}{.1in}}
\newtheorem{theorem}{Theorem}

\newtheorem{case}{Case}

\newtheorem{condition}{Condition}

\newtheorem{definition}{Definition}

\newtheorem{proposition}[theorem]{Proposition}

\newenvironment{proof}[1][Proof]{\noindent\textbf{#1.} }{\ \rule{0.5em}{0.5em}}
\begin{document}

\title{\textbf{How Macroeconomists Lost Control of Stabilization Policy: Towards Dark
Ages}\thanks{We thank two anonymous referees and editors Hans Michael
Trautwein Andr\'{e} Lapidus, Jean-S\'{e}bastien Lenfant and Goulven Rubin. We
thank Alain\ Raybaut our discussant as well as participants of the session at
ESHET\ conference in Lille 2019. We thank Marwan Simaan and Edward Nelson for
useful comments. }}
\author{Jean-Bernard Chatelain\thanks{Paris School of Economics, Universit\'{e} Paris
I Pantheon Sorbonne, 48 Boulevard Jourdan 75014 Paris. Email:
jean-bernard.chatelain@univ-paris1.fr.} \ and Kirsten Ralf\thanks{ESCE
International Business School, INSEEC\ U. Research Center, 10 rue Sextius
Michel, 75015 Paris, Email: Kirsten.Ralf@esce.fr}}
\date{September 25, 2020}
\maketitle

\begin{abstract}
This paper is a study of the history of the transplant of mathematical tools
using negative feedback for macroeconomic stabilization policy from 1948 to
1975 and the subsequent break of the use of control for stabilization policy
which occurred from 1975 to 1993.\ New-classical macroeconomists selected a
subset of the tools of control that favored their support of rules against
discretionary stabilization policy. The Lucas critique and Kydland and
Prescott's time-inconsistency were over-statements that led to the
\textquotedblleft dark ages" of the prevalence of the
stabilization-policy-ineffectiveness idea. These over-statements were later
revised following the success of the Taylor (1993) rule.

\textbf{JEL\ classification numbers}: C61, C62, E43, E44, E47, E52, E58.


\textbf{Preprint:
\href{https://doi.org/10.1080/09672567.2020.1817119}{https://doi.org/10.1080/09672567.2020.1817119}%
}

\textit{The European Journal of the History of Economic Thought} (2020), 27(6).

\end{abstract}

\section{Introduction}

This paper presents a longitudinal study of the transplant of key ideas and
mathematical tools from negative-feedback control in engineering and applied
mathematics to macroeconomic stabilization policy. This movement evolved
parallel to the \textquotedblleft rules \emph{versus}
discretion\textquotedblright\ or \textquotedblleft stabilization policy
ineffectiveness\textquotedblright\ controversy from 1948 to 1993. In
particular, we observe a fast transplant of classic control and optimal
control to stabilization policy in the 1950s and 1960s, followed by a long
delay to transplant robust control and stochastic optimal control to optimal
state estimation and optimal policy. The paper re-evaluates the Lucas critique
and time-inconsistency argument which contributed to the bifurcation with
diverging paths between control \emph{versus} the modeling of stabilization
policy by mainstream macroeconomics in the 1970s and 1980s.

Adam Smith (1776) believed that demand and supply are always self stabilizing
due to a negative feedback mechanism in the private sector. In the 1930s
emerged Keynesian macroeconomic stabilization policy where the policy maker
uses negative feedback mechanisms with monetary or fiscal policy instruments.
Friedman (1948) started the rules \emph{versus} discretion controversy, when
he proposed fiscal rules that do not vary in response to cyclical fluctuations
in business activity, so that it is the private sector's negative-feedback
mechanism that stabilizes the economy, and not the policy maker. He defined
\textquotedblleft discretion\textquotedblright\ as Keynesian state contingent
policy where the policy instruments change with respect to the deviation of
policy targets from their set points. Discretionary policy was used in the
US\ in the 1950s and 1960s, and the proponents of rules were losing in the
controversy during these two decades.

In parallel to the development of Keynesian stabilization policy, the field of
applied mathematics and engineering developed the tools of classic control
during 1930-1955 (Bennett (1996)). This field gains maturity and autonomy
while creating a world association with a first IFAC\ conference in 1960.
Between 1955\ and 1990, it has an impressive and fast rate of new discoveries:
optimal control, optimal state estimation with Kalman filter, stochastic
optimal control, robust optimal control, Nash and Stackelberg dynamic games.
These discoveries were readily applied for many devices with numerical
algorithms using the development of computers at the time. Astr\"{o}m and
Kumar (2014) survey the research field of control which is based on
negative-feedback rules stabilizing a dynamic system:

\begin{quotation}
Feedback is an ancient idea, but feedback control is a young field... Its
development as a field involved contributions from engineers, mathematicians,
economists and physicists. It represented a paradigm shift because it cut
across the traditional engineering disciplines of aeronautical, chemical,
civil, electrical and mechanical engineering, as well as economics and
operations research. The scope of control makes it the quintessential
multidisciplinary field. (Astr\"{o}m and Kumar (2014), p.~3)
\end{quotation}

There was a strong demand for control tools for firm-level planning and for
macroeconomic stabilization policy in the 1960s (Kendrick (1976), Kendrick
(2005), Neck (2009) and Turnovsky (2011)). Barnett describes the related
controversies surrounding the model of the Federal Reserve Board in the 1970s
(Barnett and Serletis (2017)):

\begin{quotation}
The policy simulations were collected together to display the policy target
paths that would result from various choices of instrument paths. The model
was very large, with hundreds of equations. Some economists advocated
replacing the \textquotedblleft menu\textquotedblright\ book of simulations
with a single recommended policy, produced by applying optimal control theory
to the model. The model was called the FMP model, for Federal
Reserve-MIT-Penn, since the origins of the model were with work done by Franco
Modigliani at MIT and Albert Ando at the U. of Pennsylvania, among others.
That model's simulations subsequently became an object of criticism by
advocates of the Lucas Critique. The alternative optimal control approach
became an object of criticism by advocates of the Kydland and Prescott (1977)
finding of time inconsistency of optimal control policy. (Barnett and Serletis
(2017), p.~7-8)
\end{quotation}

The Lucas (1976) critique and Kydland and Prescott's (1977) time-inconsistency
argument convinced a sufficient number of macroeconomists that, \emph{although
negative-feedback mechanism and optimal control should be used for modeling
the private sector, negative-feedback mechanism and optimal control cannot be
used for stabilization policy by macroeconomists and practitioners of monetary
and fiscal policy}.

In the 1970s, Lucas, Kydland and Prescott, labeled as new classical
macroeconomists, took sides for rules in the rules \emph{versus} discretion
controversy. On the one hand, incorporating control tools into macroeconomics
was obviously a scientific progress, and new classical macroeconomists
invested heavily, like Keynesian macroeconomists, in learning these tools in
the 1960s. On the other hand, control tools were supporting negative-feedback
mechanism driven by policy makers, hence discretion. Therefore, the tools of
control were pivotal in the rules \emph{versus} discretion controversy.

Efficient multi-disciplinary research tools can be imported from one field of
research to another one. However, the scientists in the field of arrival are
free to bias their choice of the tools to be imported from the field of
origin, if they are taking side in a scientific controversy. This selection
bias entails the risk of the inconsistency of the imported subset of tools
with respect to the field of origin.

The new-classical economists put forward their normative rational expectations
theory with theoretical demonstrations using a Kalman filter. They claimed it
is impossible to estimate parameters of the transmission mechanism when there
is reverse causality of the feedback rule. The new-classical economists
suggested importing time-inconsistency into dynamic games. They simulated
models using the linear quadratic regulator for the private sector. All these
approaches are using tools from the field of control.

But, in addition, following the complete guidelines of the field of control
where the accuracy of the measurement of the transmission mechanism is a key
element, they could have attempted the falsification of their theory
estimating parameters with a Kalman filter. They could have attempted to
devote a lot of resources to identification strategies when facing reverse
causality in systems of equations. They could have determined optimal feedback
policy based on these estimation of state variables using stochastic optimal
control. They could have searched for a policy that could be robust to some
ranges of uncertainty on parameters of the transmission mechanism.

But using these tools would have been inconsistent with their research agenda
where they were taking sides in a scientific controversy. They biased the use
of some tools of control in order to support their prior view on the side of
\textquotedblleft rules\textquotedblright\ in the controversy. The temporary
success (for two decades) of their selection bias restricted the demand and
delayed the use of the tools of control for stabilization policy.

This helps to reconsider how, in order to convince a sufficiently large subset
of the community of macroeconomists, \emph{the authors rhetorically
generalized some valid statements of their papers stretching them to extreme
conclusions, which, in turn, were false statements}.

The Lucas critique (1976) over-stated that it is impossible to identify
parameters in a dynamic system of equations with reverse causality. Despite
the Lucas critique, Kydland and Prescott (1982) over-stated that the
US\ economy during 1950-1979\ behaved \emph{as if} stabilization policy had
never been implemented (policy instruments were pegged) or \emph{as if} the
policy instruments did not have an effect on policy targets. Kydland and
Prescott (1977) over-stated that policy maker's credibility can never be
achieved using negative-feedback rules according to optimal control.

Simulations using private sector micro-economic foundations and
auto-regressive exogenous shocks were rhetorically presented as the magical
\textquotedblleft scientific\textquotedblright\ solution to answer the Lucas
critique, to avoid time inconsistency and to describe business cycles data.

These rhetorics marginalized or delayed for at least a decade attempts to
model stabilization policy transplanting the new tools of robust optimal
control facing parameter uncertainty and stochastic optimal control with
feedback rules reacting to estimates of state variables using a Kalman filter
with macroeconomic time-series. This outcome is labeled \textquotedblleft dark
ages\textquotedblright\ by Taylor (2007):

\begin{quotation}
But after this flurry of work in the late 1970s and early 1980s, a sort of
\textquotedblleft dark age\textquotedblright\ for this type of modeling began
to set in. Ben McCallum (1999) discussed this phenomenon in his review
lecture, and from the perspective of the history of economic thought, it is an
interesting phenomenon. As he put it, there was \textquotedblleft a long
period during which there was a great falling off in the volume of
sophisticated yet practical monetary policy analysis". (Taylor (2007))
\end{quotation}

In order to explain the selection of tools imported from control on behalf of
the new-classical macroeconomist side in the \textquotedblleft rules
\emph{versus} discretion\textquotedblright\ controversy, our method is to use
as a reference model the simplest model of control. We translate the
mathematical arguments of the most cited papers in this controversy into the
framework of this single model. This helps to understand how different
definitions of discretion and how different hypothesis on the persistence of
the policy targets in the transmission mechanism matter, even though they were
not highlighted so far in the history of the rules \emph{versus} discretion controversy.

The structure of this study is as follows:

Section 2\ frames the original classical economists' view of self-stabilizing
markets in the framework of Ezekiel's (1938) Cobweb model. Friedman's (1948)
and Kydland and Prescott's (1977) rules \emph{versus} discretion controversy
is presented in the framework of the simplest first order single-input
single-output model of control.

Section 3 documents the fast transplant of classic control from Phillips
(1954b) and resurrection with the Taylor (1993) rule. It mentions the fast
transplant of optimal control to stabilization policy in the 1960s and of a
Kalman filter for rational expectations theory. By contrast, It mentions the
very limited use of Stochastic Optimal Control using simultaneously
Kalman-filter estimations for determining optimal stabilization policy in the
linear quadratic Gaussian model. Finally, it emphasizes the long delay before
transplanting robust control, dealing with the uncertainty on parameters.

Section 4\ re-evaluates the claim of the Lucas (1976) critique that it is
impossible to identify the parameters of the transmission mechanism when there
is reverse causality due to a negative-feedback rule. The Lucas critique is
not resolved by microeconomic foundations, by Sims' (1980) vector
autoregressive models, by Kydland and Prescott's (1982) real business cycle,
nor by Lucas' (1987) welfare cost of business cycles. Kydland and Prescott
(1977) section 5 has a specific definition of discretion assuming a Lucas
critique bias, which is not related to time-inconsistency.

Section 5 re-evaluates the time-inconsistency argument and the impossibility
of policy maker's credibility leading to the impossibility to use negative
feedback grounded by control. Firstly, it credits time-inconsistency to Simaan
and Cruz' (1973b) first contribution with Kydland (1975, 1977), Calvo (1978)
and Kydland and Prescott (1980) as followers. Secondly, it highlights that the
inflationary bias in Barro and Gordon's (1983) static model is distinct from
Calvo's (1978) dynamic time-inconsistency.

Section 6 explains how Taylor (1993) rhetorically translated Friedman's (1948)
rules \emph{versus} discretion controversy in his \textquotedblleft
Semantics\textquotedblright\ section to the advantage of the negative-feedback
mechanism of his Taylor rule.

Section 7 concludes that the \textquotedblleft rules \emph{versus}
discretion\textquotedblright\ controversy on macroeconomic stabilization
policy biased and delayed the efficient transfers of knowledge from another
field of research (the field of control), and by doing so, it delayed
scientific progress.

\section{Rules \emph{versus} Discretion and Control}

\subsection{Self-Stabilizing Markets: Smith (1776) and Ezekiel's (1938)
Counter-Example}

The main underlying disagreement of the debate starting with Friedman's (1948)
`rules \emph{versus} discretion' and continuing with the new-classical
macroeconomists' attack against stabilization policy during the 1970s and
1980s, is the question which forces stabilize the economy. Relating this to
optimal control, i.e.~finding a control law such that an objective function is
optimized in a dynamical system, the question is whether it is the private
sector's behavior alone that leads to an economic equilibrium or whether there
is a need for economic policy in the form of government intervention. As a
prominent example, optimal control of the private sector, using negative
feedback, stabilizes the markets in Kydland and Prescott's (1982) business
cycle model.

As has been noted by Mayr (1971) it is even possible to interpret Adam Smith's
(1776) self-regulating local stability of supply and demand market equilibrium
as a negative-feedback mechanism.

\begin{quotation}
When the quantity brought to market exceeds the effectual demand, it cannot be
all sold to those who are willing to pay the whole value.... Some part must be
sold to those who are willing to pay less, and the low price which they give
for it must reduce the price of the whole. (Adam Smith's (1776, chapter\ 7))
\end{quotation}

Mayr (1971), however, did not relate his control translation of Smith to the
Cobweb model (Ezekiel (1938)). In the Cobweb model, the deviation of the
market price $p_{t}$ from its natural (equilibrium) price $p^{\ast}$ is a
decreasing function of excess supply, the difference between supply $x_{t}%
^{s}$ and demand $x_{t}^{d}$:%

\[
p_{t}-p^{\ast}=F\left(  x_{t}^{s}-x_{t}^{d}\right)  \text{ with }F<0.
\]

Conversely, the difference between the market price and the natural price is
the signal that tells the producer whether to increase or decrease his
production. Excess supply increases with the price, including a time lag to
adjust supply:%

\[
x_{t+1}^{s}-x_{t+1}^{d}=B\left(  p_{t}-p^{\ast}\right)  \text{ with }B>0.
\]

Excess supply does not depend on its own lagged value $x_{t}^{s}-x_{t}^{d}$.
There is no persistence of excess supply in the case where the price is set to
its equilibrium value: $p_{t}=p^{\ast}$.

\begin{quotation}
If at any time it [the supply] exceeds the effectual demand, some of the
component parts of its price must be paid below their natural rate. If it is
rent, the interest of the landlords will immediately prompt them to withdraw a
part of their land; and if it is wages or profit, the interest of the
labourers in the one case, and of their employers in the other, will prompt
them to withdraw a part of their labour or stock from this employment. The
quantity brought to market will soon be no more than sufficient to supply the
effectual demand. (Adam Smith's (1776, chapter\ 7))
\end{quotation}

Smith concludes that this feedback mechanism implies that market prices tends
towards the natural equilibrium price:

\begin{quotation}
The natural price, therefore, is, as it were, the central price, to which the
prices of all commodities are \textbf{continually gravitating.} Different
accidents may sometimes keep them suspended a good deal above it, and
sometimes force them down even somewhat below it. But whatever may be the
obstacles which hinder them from settling in this center of repose and
continuance, they are \textbf{constantly tending towards it}. (Adam Smith
(1776, chapter\ 7))
\end{quotation}

As opposed to Smith's (1776) intuition, however, the convergence result in the
Cobweb dynamics is only valid under a specific condition for price
elasticities of supply and demand:%

\[
x_{t+1}^{s}-x_{t+1}^{d}=BF\left(  x_{t}^{s}-x_{t}^{d}\right)  \text{ requires
}-1<BF<0.
\]

Feedback can bring local stability (negative feedback) or local instability
(positive feedback) within the private sector. Although Ezekiel (1938) does
not cite Smith (1776), he uses the same word (\textquotedblleft
gravitate\textquotedblright) for describing classical economic theory:

\begin{quotation}
Classical economic theory rests upon the assumption that price and production,
if disturbed from their equilibrium \textbf{tend to gravitate back toward that
normal}. The cobweb theory demonstrates that, even under static conditions,
this result will not necessarily follow. On the contrary, prices and
production of some commodities might tend to fluctuate indefinitely [case
$BF=-1$], or even to diverge further and further from equilibrium. [case
$BF<-1$] (Ezekiel (1938), p.~278-279)
\end{quotation}

When $BF=0$ ($F=0$ or $B=0$) the adjustment towards the equilibrium following
an excess supply or excess demand shock takes only one period. The
demand-first equation can be interpreted as a proportional feedback rule,
where the price plays the role of the feedback-policy instrument of the
private sector. Since the Cobweb model assumes zero open-loop persistence of
supply, if the price elasticity of demand and therefore the parameter $F$ was
to be chosen optimally, it would be set to an infinite elasticity ($F=0$).
Then the optimal policy of the private sector is to peg the price at its
optimal value $p_{t}=p^{\ast}$.

\subsection{Positive \emph{versus} Negative Feedback in a First-Order
Two-Inputs Single-Output Linear Model}

Before going into details on the rules-\emph{versus}-discretion controversy in
the next section, we will clarify the concepts of positive and negative
feedback and their relation to controllability and local stability. For this
we introduce explicitly a policy maker and consider the monetary policy
transmission mechanism as a \textquotedblleft\emph{first-order two-inputs
single-output}\textquotedblright\ linear model as used in dynamic games
(Simaan and Cruz (1973a)). First-order stands for one lag of the policy target
in the transmission mechanism. The first input is a policy instrument decided
by the private sector (for example output or consumption $x_{t}$). The second
input is a policy instrument decided by the policy maker (for example, nominal
funds rate $i_{t}$). The single output or single policy target can be
inflation $\pi_{t}$ (instead of the price level $p_{t}$ in the cobweb model).
The policy target and the policy instruments are are written in deviation of
their long run equilibrium values:%

\begin{equation}
\pi_{t+1}=A^{\prime}\pi_{t}+B^{\prime}x_{t}+Bi_{t}+\varepsilon_{t}\text{ with
}A^{\prime}\geq0\text{, }B^{\prime}\neq0\text{, }B\neq0\text{, }\pi_{0}\text{
given.}%
\end{equation}
Additive disturbances are denoted $\varepsilon_{t}$ and are assumed to be
identically and independently distributed. If $B^{\prime}\neq0$, the private
sector's policy instrument is correlated with the future value of the policy
target. Then, this first-order linear model is Kalman (1960a) controllable by
the private sector. If $B\neq0$, the policy maker's policy instrument is
correlated with the future value of the policy target. Then, this first-order
linear model is Kalman (1960a) controllable by the policy maker.\footnote{As a
reminder, a model exhibits Kalman controllability if the policy instruments
have a direct or indirect effect on the policy target.} Both, the private
sector and the policy maker behave according to proportional feedback rules
given by:%

\begin{equation}
x_{t}=F^{\prime}\pi_{t}\text{ and }i_{t}=F\pi_{t}\text{ with }F^{\prime}\in%
\mathbb{R}
\text{ and }F\in%
\mathbb{R}
\text{. }%
\end{equation}
In a first step, substituting the private sector's feedback rule in the
transmission mechanism implies:%
\begin{equation}
\pi_{t+1}=A\pi_{t}+Bi_{t}+\varepsilon_{t}\text{ where }A=A^{\prime}+B^{\prime
}F^{\prime}\text{ .}%
\end{equation}
For $A^{\prime}$ and $B^{\prime}$ given and if the values of the policy
instrument $x_{t}$ are not constrained (for example by a endowment
constraint), the private sector can choose any real value for $F^{\prime}$ and
therefor also for $A^{\prime}+B^{\prime}F^{\prime}$. As a consequence
$E_{t}\pi_{t+1}$ can take any target value.

Consider the case where the policy maker pegs its policy instrument to its
long run value ($i_{t}=0$). As we never measure negative auto-correlation for
macroeconomic time series, we can assume that $A\geq0$.\ We also never measure
zero auto-correlation (no persistence) for macroeconomic time series.
Nonetheless, we also consider the case of zero persistence $A=0$ in this
paper, because it played a crucial, but unnoticed role in the rules
\emph{versus} discretion controversy. Then, three outcomes are possible for
the private sector's feedback:

\begin{case}
$A=A^{\prime}+B^{\prime}F^{\prime}=0$ for $F^{\prime}=\frac{-A^{\prime}}{B}$.
The value of the private sector's feedback rule parameter $F^{\prime}%
$\ implies no persistence of the policy target following a random shock
without persistence.
\end{case}

\begin{case}
$0<A=A^{\prime}+B^{\prime}F^{\prime}<1$ for $\frac{-A^{\prime}}{B}<F^{\prime
}<\frac{1-A^{\prime}}{B}$ if $B^{\prime}>0$. The value of the feedback rule
parameter $F^{\prime}$\ implies persistence with stationary dynamics of the
policy target following a random shock without persistence.
\end{case}

\begin{case}
$A=A^{\prime}+B^{\prime}F^{\prime}\geq1$ for $F^{\prime}>\frac{1-A^{\prime}%
}{B}$ if $B^{\prime}>0$. The value of the feedback rule parameter $F^{\prime}%
$\ implies a diverging trend with non-stationary dynamics of the policy target
following a random shock without persistence.
\end{case}

Negative feedback and positive feedback mechanism are defined in the following way:

\begin{definition}
Negative-feedback rule parameters $F^{\prime}$\ are such that $0\leq
A^{\prime}+B^{\prime}F^{\prime}<A^{\prime}$, which implies $B^{\prime
}F^{\prime}<0$.
\end{definition}

Since any disturbance automatically causes corrective action in the opposite
direction, the parameters $B$ and $F$ have opposite signs.

\begin{definition}
Positive-feedback rule parameters $F$ are such that $0\leq A^{\prime
}<A^{\prime}+B^{\prime}F^{\prime}$, which implies $B^{\prime}F^{\prime}>0$.
\end{definition}

\begin{proposition}
Negative feedback does not imply local stability for the private sector's
policy-rule parameters $F^{\prime}$ such that $0\leq1<A^{\prime}+B^{\prime
}F^{\prime}<A^{\prime}$. The requirement for negative feedback and local
stability is therefore that the private sector's policy-rule parameter
$F^{\prime}$ satisfies: $0\leq A^{\prime}+B^{\prime}F^{\prime}<\min\left(
A^{\prime},1\right)  $. Conversely, positive feedback does not imply local
instability for the private sector's policy-rule parameters $0\leq A^{\prime
}<A^{\prime}+B^{\prime}F^{\prime}<1$.
\end{proposition}

In a second step, the policy maker chooses $A=A^{\prime}+B^{\prime}F^{\prime}$
and $B$. Substituting the policy maker's feedback-rule in the transmission
mechanism implies:%

\begin{equation}
\pi_{t+1}=(A+BF)\pi_{t}+\varepsilon_{t}\text{. }%
\end{equation}

For $A$ and $B\neq0$ given and if the values of the policy instrument are not
constrained (for example by a zero lower bound for funds rate) so that the
policy maker can choose any real value for $F$, the policy maker can target
$E_{t}\pi_{t+1}$ at any real value because he can choose any real value $A+BF$.

The condition for the policy maker's negative-feedback and stabilizing
policy-rule parameters is given by the set of parameters $F$ satisfying:
$0\leq A+BF<\min\left(  A,1\right)  $.

\subsection{Rules \emph{versus} Discretion}

Even though the main focus of the paper is the period after 1948, it is worth
mentioning Simons' (1936) article on \textquotedblleft rules versus
authorities\textquotedblright\ as a predecessor of the literature on policy
rules. According to Simons (1936), authorities are not necessarily only
related to a policy maker's negative-feedback behavior, but also to any random
or erroneous policy maker's decision, such as Gold standard \textquotedblleft
rules\textquotedblright\ or trade wars. First of all, a 100\% reserve
requirement should be set so that private banks and shadow banks cannot create
money. This would eliminate the financial instability due to a banking crisis.
Secondly, a rigid public rule on central bank public creation of money should
be fixed. Thirdly, free market competition in the real sector should prevail.

Based on these ideas, Friedman (1948) defines rules \emph{versus} discretion
as follows:

\begin{quotation}
[For government expenditures excluding transfers,] no attempt should be made
to vary expenditures, either directly or inversely, in response to cyclical
fluctuations in business activity. ... The [transfer] program [such as
unemployment benefits] should not be changed in response to cyclical
fluctuations in business activity. Absolute outlays, however, will vary
automatically over the cycle. They will tend to be high when unemployment is
high and low when unemployment is low. (Friedman (1948), p.~248)
\end{quotation}

In addition to Simons' (1936) 100\% reserve requirements by private financial
institutions, Friedman (1948) also advocated zero public debt and allowed
money supply to move cyclically in order to finance cyclical deficits or
surpluses. He shifted to a fixed money-supply growth-rate rule in Friedman
(1960): \textquotedblleft\emph{The stock of money [should be] increased at a
fixed rate year-in and year-out without any variation in the rate of increase
to meet cyclical needs}\textquotedblright. Rules can be Friedman's (1960)
fixed $k$-percent growth rate of money supply $(m_{t}=0)$, no public deficits
$(s_{t}=0)$, an interest rate peg $(i_{t}=0)$, or an exchange rate peg
$(e_{t}=0)$. These definitions are the same as in Kydland and Prescott (1977).
They will change, however, at the end of the 1970s (see section 6). Using our
framework, rules are defined as follows:

\begin{definition}
A policy maker follows a \textquotedblleft Rule\textquotedblright\ whenever he
pegs his policy instruments to their steady state values ($F=0$ and $i_{t}%
=0$), with policy target dynamics $\pi_{t+1}=A\pi_{t}+\varepsilon_{t}$.
\end{definition}

\begin{condition}
In order to have stable dynamics for the policy target with \textquotedblleft
Rules\textquotedblright, the private sector always decides to stabilize the
value of the policy-rule parameter $F^{\prime}$: $0\leq A^{\prime}+B^{\prime
}F^{\prime}<1$ (case 1\ and case 2).
\end{condition}

One could interpret this behavior of the private sector as Smith's (1776)
implicit hypothesis of market clearing.

Negative-feedback counter-cyclical fiscal policy evolved with Keynes' (1936)
``General Theory'' and the idea that the equilibrium is not automatically
reached by market forces alone, but that there exist situations where a
government intervention is necessary.

\begin{definition}
\textquotedblleft Discretion\textquotedblright\ is a policy that responds to
cyclical fluctuations of the deviations of the policy variables from their
long run target values $(i_{t}=F\pi_{t}$ with $F\neq0)$, with policy target
dynamics $\pi_{t+1}=\left(  A+BF\right)  \pi_{t}+\varepsilon_{t}$.
\end{definition}

\begin{condition}
In order to have stable dynamics for the policy target and policy maker's
negative-feedback, the policy maker decides policy rule parameter $F$ to
satisfy $0<A+BF<A<1$ in case 2, or to satisfy $0<A+BF<1<A$ in case 3. In case
1, $0\leq A+BF\leq A=0$, a discretionary policy with $BF\neq0$ adds
persistence with respect to a policy pegging the policy instrument to its long
run value ($F=0$): negative-feedback cannot be achieved with $F\neq0$.
\end{condition}

Using the above framework, table 1\ summarizes distinctive features of the
rules \emph{versus} discretion (1948-1993) controversy.

\textbf{Table 1: Rules \emph{versus} Discretion, Stabilization Policy
Ineffectiveness \emph{versus} Negative-feedback Controversy (1948-1993).}%

\begin{tabular}
[c]{|c|c|c|}\hline
& \textbf{\textquotedblleft Rules\textquotedblright, \textquotedblleft
Laissez-faire\textquotedblright} & \textbf{\textquotedblleft
Discretion\textquotedblright}\\\hline
Policy & Ineffective & Effective\\\hline
Proponents & Friedman, Barro, Gordon & Phillips, Taylor\\\hline
Fischer's label & Inactive (Passive) & Activist\\\hline
Feedback & No feedback, Peg & Negative feedback\\\hline
\textquotedblleft Rule\textquotedblright & $i_{t}=F\pi_{t}=0$ with $F=0$ &
$i_{t}=F\pi_{t}$ with $F\in D_{nf}$\\\hline
Transmission & $\pi_{t+1}=A\pi_{t}+Bi_{t}+\varepsilon_{t}$ & $\pi_{t+1}%
=A\pi_{t}+Bi_{t}+\varepsilon_{t}$\\\hline
Controllability & Not necessary: $B\in%
\mathbb{R}
$ & Necessary: $B\neq0$\\\hline
Optimal & $A=0$ & $A>0$\\\hline
Phillips curve & Static: $A=0$ & Accelerationist: $A=1$\\\hline
Dynamics & $E_{t}\pi_{t+1}=A^{t}\pi_{0}$ & $E_{t}\pi_{t+1}=\left(
A+BF\right)  ^{t}\pi_{0}<A^{t}\pi_{0}$\\\hline%
\begin{tabular}
[c]{c}%
Target $\pi_{t}$\\
persistence\\
condition
\end{tabular}
&
\begin{tabular}
[c]{c}%
Private sector\\
necessarily stationary\\
dynamics: $0\leq A<1$%
\end{tabular}
&
\begin{tabular}
[c]{c}%
$0\leq A+BF<\min\left(  A,1\right)  $\\
$BF<\min\left(  0,1-A\right)  $\\
$F$ opposite sign of $B$%
\end{tabular}
\\\hline
Lucas critique & $\ A$ not a function of $F$ & $A+BF$ function of $F$\\\hline
Time inconsistent & Static model: irrelevant & Time inconsistent if $\pi_{t}$
jumps\\\hline
Control label & Open loop & Closed loop\\\hline
\end{tabular}
\bigskip

Prominent supporters of the \textquotedblleft rules\textquotedblright\ side
are Friedman (1948, 1960) and Barro and Gordon (1983a and b). Their implicit
assumption is case 1 ($A=0$) and static models without lags of the policy
target (Blanchard and Fischer (1989), p.581). Barro and Gordon (1983a and b)
consider a static Phillips curve ($A=0$).

Prominent supporters of the \textquotedblleft discretion\textquotedblright%
\ side are Phillips (1954b), Taylor (1993, 1999). They implicitly assume a
dynamic model with trend (case 3: $A>1$) and possibly with stationary
persistence (case 2: $0<A<1$). In case 2, persistence $A$\ is not too small so
that the reduction of persistence subtracting $BF<0$ is not negligible. Taylor
(1999) and Fuhrer (2010) consider an accelerationist Phillips curve ($A=1$),
which may be related to trend inflation in the 1970s in the USA. Volcker's
discretionary monetary policy during 1979-1982 reversed non-stationary trend
inflation into stationary inflation $0\leq A+BF<1\leq A$. The rule parameter
$F$\ is a bifurcation parameter, for given values of the parameters $A$ and
$B$ of monetary policy transmission mechanism.

In control theory, the policy responses are always conditional on the
transmission mechanism. It does not make sense to put forward a policy rule
without specifying the policy transmission mechanism. For Nelson (2008, p.95),
Friedman and Taylor agreed "\emph{on the specification of shocks, policy
makers' objectives and trade-offs. Where they differed was on the extent to
which structural models should enter the monetary policy decision making
process.}"

\section{The Transplants of Control Tools from Engineering to Stabilization
Policy}

Smith did not refer to engines as an analogy when he explained the
negative-feedback mechanism in the \textit{Wealth of Nations} (see Mayr
(1971)) even though he knew Watt and engineers' machines using
negative-feedback. The concept of negative feedback was used by economist most
of the time without a reference to engineer's techniques until classic control
emerged. With respect to the modeling of macroeconomic stabilization policy
during the period of the rules \emph{versus} discretion controversy 1948-1993,
there have been three stages of implementation of control theory according to
Zhou, Doyle and Glover (1996) and Hansen and Sargent (2008). The first one is
classic control without a loss function in the 1950s, with proportional,
integral and derivative (P.I.D) policy rules. The second one is optimal
control including a quadratic loss function with a Kalman linear quadratic
regulator, optimal state estimation with a Kalman filter and stochastic
optimal control merging both methods with the linear quadratic Gaussian model
in the 1960s. The third stage is robust control which takes into account
uncertainty on the parameters of the policy transmission mechanism in the 1980s.

\subsection{The Fast Transplant of Classic Control}

Tustin (1953), an electrical engineer at the University of Birmingham,
mentions to have started in 1946 applying classic control methods used in
electrical systems to Keynesian macro-models (Bissell (2010)). Phillips
(1954a), an electrical engineer hired at the London School of Economics, wrote
a two page book review on Tustin (1953) in the Economic Journal. He built the
hydraulic computer MONIAC (Monetary National Income Analogue Computer) in 1949
(Leeson (2011)). The MONIAC is a series of connected glass tubes filled with
water where the flow represented GNP and the feedback system represented the
use of monetary and fiscal policy. Phillips (1954b (from Phillips' PhD) and,
1957) used proportional, integral and derivative (P.I.D.) rules of classic
control to stabilize an economic model using negative-feedback mechanism
(Hayes (2011)). Taylor's (1968) master thesis merged Phillips' (1961) model of
cyclical growth with Phillips' (1954b) proportional, integral and derivative
negative-feedback stabilization rules. Thirty-nine years after Phillips
(1954b) and twenty-five years after Taylor's (1968) P.I.D rules, the Taylor
(1993) rule is a proportional (P) feedback rule of classic control. Taylor, in
Leeson and Taylor (2012)), explains why he took sides with discretion:

\begin{quotation}
I viewed policy rule as a natural way to evaluate policy in the kinds of
macroeconomic models which I learned and worked on at Princeton and Stanford.
It was more practical than philosophical or political. (Leeson and Taylor (2012))
\end{quotation}

We now highlight how classic control takes sides with \textquotedblleft
discretion\textquotedblright: The policy maker targets his preferred
persistence of the time-series of policy targets, which determines his
preferred speed of convergence of these policy targets to their long run
equilibrium. For example, a central bank could do an inflation persistence
targeting of an auto-correlation of inflation of $0.8$, satisfying a stability
and negative feedback condition: $0\leq\lambda^{\ast}=A+BF^{\ast}%
=0.8<\min(A,1)$. This decision is called in classic control \textquotedblleft
pole placement\textquotedblright, because $\lambda$ is a pole or a zero of the
polynomial at the denominator of the Laplace transform of the closed-loop system.

Accordingly, the policy maker decides on a policy rule parameter $F^{\ast
}=\frac{\lambda^{\ast}-A}{B}=\frac{0.8-A}{B}$ in the case of a proportional
feedback rule. For example, if there is a negative marginal effect of the
funds rate on inflation ($B<0$), the policy rule $F^{\ast}$ is an affine
decreasing function of the inflation persistence target $\lambda^{\ast}$.

Taylor's (1999) transmission mechanism is an accelerationist Phillips curve
(where $x_{t}$ is the output gap and $a$ is the slope of the Phillips curve)
and an investment saving (IS) equation:%

\begin{equation}
\pi_{t+1}=\pi_{t}+ax_{t}\text{, }a>0\text{ and }x_{t}=-b(i_{t}-\pi_{t})\text{,
}b>0.
\end{equation}

So that the transmission mechanism of monetary policy is such as $A=1+ab>1$:%

\begin{equation}
\pi_{t+1}=\left(  1-B\right)  \pi_{t}+Bi_{t}\text{ with }B=-ab<0.
\end{equation}

The Taylor principle states that the funds rate should respond by more than
one to deviation of inflation from its long run target ($F>1$). The Taylor
principle corresponds to the classic control condition for negative-feedback
rule parameters such that $0<A+BF<\min(1,A)$, for models such that $B<0$ and
$A=1-B>1$ (Taylor (1999)):%

\begin{equation}
0\leq A+BF=1-B+BF<1\text{ and }B<0\Rightarrow1<F<-\frac{B}{A}=-\frac{B}{1-B}.
\end{equation}

The upper bound condition on $F$ corresponds to zero persistence of inflation.

\subsection{The Fast Transplant of Optimal Control}

The second step in the transfer of methods used in engineering to economic
modeling is the introduction of optimal control in the 1950s, see Duarte
(2009) and Klein (2015). In optimal control, a quadratic loss function is
minimized subject to linear dynamic equations. Using a certainty-equivalence
property, normal disturbances with zero mean can be added to the model
according to Simon (1956) and Theil (1957).

In the 1950s, Simon, Holt, Modigliani and Muth came to the \textit{Graduate
School of Industrial Administration} at the \textit{Carnegie Institute of
Technology} in Pittsburgh. Holt had come from an engineering background at
M.I.T. and Simon's father was an electrical engineer after earning his
engineering degree in Technische Hochschule Darmstadt.\footnote{Phillips' and
Taylor's fathers were also engineers, (Leeson and Taylor (2012)} (Leeson and
Taylor (2012)). They applied control methods to microeconomics by computing
variables for production, inventories and the labor force of a firm. Optimal
linear decision rules from linear quadratic models were computed for specific
economic models of firm's production by Holt, Modigliani and Simon (1955) and
Holt, Modigliani and Muth (1956), (Singhal and Singhal (2007)). Holt (1962)
developed an optimal control model to analyze fiscal and monetary policy.

Kalman (1930-2016), an electrical engineer, wrote the key paper for solving
linear quadratic optimal control (linear quadratic regulator,
LQR).\footnote{He was invited to participate in the world econometric congress
in 1980 in Aix en Provence, but his paper was not published in Econometrica.
He later received the highest honor of US science, the US\ medal of science in
2009.} He extended the static Tinbergen (1952) principle, namely that there
should be as many policy instrument as policy targets, to a dynamic setting.
Kalman's (1960a) controllability definition is such that a single instrument
can control for example three policy targets, but in three different periods
(Aoki (1975)). Masanao Aoki (1931-2018) was a Japanese professor of
engineering at UCLA and California, Berkeley from 1960-1974, before switching
fields to economics. Kalman's (1960a) linear quadratic regulator sets the
solution of stabilization policy facing a quadratic loss function solving
matrix Riccati equations. Textbooks include Sworder (1966) and Wonham (1974)
among others.

The diffusion of control techniques to macroeconomics was complementary to the
development of large scale macroeconomic models:

\begin{quotation}
Professor Bryson was offering a course in the control theory in 1966 that
caught the attention of a small group of economics graduate students and young
faculty members at Harvard... Rod Dobell, Hayne Leland, Stephen Turnovsky,
Chris Dougherty, Lance Taylor and I persisted... Two control engineers who had
shifted their interest to economics -- David Livesey at Cambridge University
and Robert Pindyck at MIT -- developed macroeconomic control theory models (in
1971 and 1972).... In May of 1972, a meeting of economists and control
engineers was arranged at Princeton University by three economists (Edwin Kuh,
Gregory Chow and M. Ishaq Nadiri) and a control engineer. The meeting which
was attended by about 40 economists and 20 engineers was to explore the
possibility that the application of stochastic control techniques, which had
been developed in engineering, would prove to be useful in economics as well
(Athans and Chow, 1972).... Another British-trained control engineer, Anthony
Healy, who was teaching at the University of Texas at the time, took an
interest in economic models and applied the use of feedback rules to a
well-known model that had been developed at the St. Louis Federal Reserve Bank
(FRB). (Kendrick (2005), p.~7-8)
\end{quotation}

In the ongoing debate ``discretion'' \emph{versus} ``rules'', optimal control
-- with a policy maker minimizing a loss function -- was used by
macroeconomists in favor of ``discretion'', whereas those in favor of
``rules'' model the private sector's optimal behavior with stationary policy
targets ($0<A=A^{\prime}+B^{\prime}F^{\prime}<1$). Nonetheless, even in this
case, optimal control by the policy maker is still able to decrease the loss
function further.

Optimal control is filling a gap in the \textquotedblleft pole
placement\textquotedblright\ method of classic control where the criterion for
choosing the persistence $\lambda^{\ast}$ of the policy target is not
explicitly stated ($0<\lambda^{\ast}=A+BF^{\ast}<1$). Kalman's optimal control
uses a quadratic loss function with the possibility of discounting future
periods with a factor $\beta$ and non-zero quadratic cost of changing the
policy instrument, $R>0$, in order to ensure concavity. The relative cost of
changing the policy instrument ($R/Q$) represents e.g. central banks'
interest-rate smoothing or governments' tax smoothing. In the case of the
private sector, it corresponds to households' consumption smoothing or firms'
adjustment costs of investment. Maximize%
\begin{equation}
-\frac{1}{2}%
{\displaystyle\sum\limits_{t=0}^{+\infty}}
\beta^{t}\left(  Q\pi_{t}^{2}+Ri_{t}^{2}\right)  \text{, with }\ R>0\text{,
}Q\geq0\text{ and }0<\beta\leq1,
\end{equation}
subject to the same transmission mechanism as the one of classic control.

Solving this linear quadratic regulator yields two roots of a characteristic
polynomial of order two, one of them stable, the other one unstable. For this
reason, Blanchard and Kahn (1980) call this solution \textquotedblleft
saddlepath stable\textquotedblright, in a space which adds co-state variables
(Lagrange multipliers) and state variables (policy targets). This implies that
the dynamics of the state variables is stable, exactly like in classic control.

Optimal persistence $\lambda^{\ast}$ is a continuous increasing function of
the relative cost of changing the policy instrument: $\lambda^{\ast}\left(
\frac{R}{Q}\right)  $. The relation between the persistence of the policy
target and policy-rule parameter, $A+BF^{\ast}=\lambda^{\ast}$, is the same
for optimal control as for classic control. Unless the targeted persistence
does not belong to the interval: $\lambda^{\ast}\in\left]  0,\min\left(
A,\frac{1}{\beta A}\right)  \right[  $, a simple rule derived from classic
control with an \emph{ad hoc} targeted persistence $\lambda^{\ast}$\ is
observationally equivalent to an optimal rule $\lambda^{\ast}\left(  \frac
{R}{Q}\right)  $, with identical predictions and behavior of the policy maker.
A reduced form of a simple rule parameter $F$ corresponds to an optimal rule
$F\left(  \frac{R}{Q}\right)  $ with preferences $\frac{R}{Q}$. There exist
preferences of the policy maker that \textquotedblleft
rationalize\textquotedblright\ an estimated value of a simple rule parameter
$F$ to be a reduced form of an optimal rule parameter $F\left(  \frac{R}%
{Q}\right)  $.

\begin{proposition}
A rule pegging the policy instrument ($i^{\ast}=0$, $F=0$) is \textbf{optimal}
for a quadratic loss function (which is bounded if: $\beta A^{2}\neq1$), with
a first-order single policy maker's instrument single-policy-target
transmission mechanism:

(i) for stationary persistence of the policy target $0<A<1/\sqrt{\beta}$ and
for a zero weight ($Q=0$) on the volatility of the policy target in the policy
maker's loss function.

(ii) for zero persistence of the policy target $A=0$ and a positive weight
$Q\geq0$ of the volatility of the policy target in the loss function.
\end{proposition}

\begin{proof}
(i) With zero weight on the policy target ($Q=0$), the relative cost of
changing the policy instrument is infinite ($R/Q\rightarrow+\infty$), which
corresponds to maximal inertia of the policy. Only the first of these two
cases allows $F^{\ast}=0$. In the second case, $\ F^{\ast}=0$ corresponds to
$\beta A^{2}=1$ with unbounded utility (see appendix).%

\begin{align*}
0  &  \leq F^{\ast}=\frac{\lambda^{\ast}-A}{B}\leq\frac{-A}{B}\text{, if
}0<A<\frac{1}{\sqrt{\beta}}\text{, }B<0\text{, }0<\beta\leq1.\\
\frac{\frac{1}{\beta A}-A}{B}  &  \leq F^{\ast}=\frac{\lambda^{\ast}-A}{B}%
\leq\frac{-A}{B}\text{, if }A>\frac{1}{\sqrt{\beta}}\text{ and }B<0\text{,
}0<\beta\leq1.
\end{align*}

If the feedback parameter is zero ($F=0$), optimal policy corresponds to a
rule ($i^{\ast}=0$). The policy target dynamics $\pi_{t}=A\pi_{t-1}$ is not
taken into account in the expected loss function.

(ii) For $A=0$ and for $B\neq0$, the system is controllable with a non-zero
effect of the policy instrument on the policy target:%
\[
Min\text{ }\left(  \beta Q\pi_{t}^{2}+Ri_{t-1}^{2}\right)  \text{ subject to:
}\pi_{t}=Bi_{t-1}\Rightarrow Min\text{ }\left(  \left(  \beta Q+R\frac
{1}{B^{2}}\right)  \pi_{t}^{2}\right)  .
\]
For each period with identical repeated optimizations, a rule ($i^{\ast}%
=0=\pi^{\ast}$) is the optimal solution, whatever the magnitude of the
relative cost of changing the policy instrument $R/Q>0$.
\end{proof}

The weakness of static analysis, however, was stated by Phillips (1954b) in a
seminal paper on stabilization:

\begin{quotation}
The time path of income, production and employment during the process of
adjustment is not revealed. It is quite possible that certain types of policy
may give rise to undesired fluctuations, or even cause a previously stable
system to become unstable, although the final equilibrium position as shown by
a static analysis appears to be quite satisfactory. (Phillips (1954a, p.290))
\end{quotation}

\textquotedblleft Rules\textquotedblright\ are optimal for a \emph{static}
model $A=0$, but in a \emph{dynamic} model where $A\geq1$, rules lead to
instability and huge welfare losses. In addition, the hypothesis $A=0$ cannot
explain the observed persistence of all macroeconomic time series, which are
stationary if $0<A<1$ or non-stationary if $A\geq1$. It cannot explain the
variations of the non-zero persistence of inflation ($A+BF>0$), such as the
acceleration of inflation or the disinflationary period during 1973\ to 1982
with observed changes of the Fed's policy behavior during Volcker's mandate.

Optimal policy (rules) in a repeated \emph{static} model of the transmission
mechanism ($A=0$, as with e.g.~a Phillips curve) is never a modeling shortcut
with results extended by analogy to optimal policy in \emph{dynamic} models of
the transmission mechanism ($A\geq1$, as with e.g.~an accelerationist Phillips
curve or a new-Keynesian Phillips curve) where discretion is optimal.

The above proposition for \textquotedblleft rules\textquotedblright\ is bad
news for the proponents of \textquotedblleft rules\textquotedblright\ because
\textquotedblleft rules\textquotedblright\ are sub-optimal in dynamic models
with policy makers using optimal control. Therefore, it was top of the
research agenda of the proponents of \textquotedblleft rules\textquotedblright%
\ in the controversy with discretion (that is, the opponents of Keynesian
stabilization policy) to seek opportunities in order to add distortions in
order to "prove" that optimal control techniques could not be used by policy
makers for stabilization policy.

In addition, the above proposition renders useless Friedman (1953) results
recently put forward by Forder and Monnery (2019) as a contribution for his
Nobel prize. He considers the simplest static model for the policy
transmission mechanism:
\[
\pi_{CL}=\pi_{OL}+i.
\]

The policy target randomly deviates from its zero equilibrium value with the
value $\pi_{OL}$ if the policy instrument $i$ is set to zero (open loop).
After a change of the policy instrument $i$, the closed loop value of the
policy target is equal to $\pi_{CL}$. The policy maker's loss function is the
volatility of the closed loop policy target $\sigma_{\pi_{CL}}^{2}$. The loss
function is obviously minimized $\sigma_{\pi_{CL}}^{2}=0$ for this optimal
negative feedback rule: $i=-\pi_{OL}$, which implies $\pi_{CL}=0$, $\rho
_{i\pi_{OL}}=-1$ and $\sigma_{i}=\sigma_{\pi_{OL}}$. Friedman's (1953)
condition for sub-optimal \emph{negative feedback} rules is obtained after
substitution of the transmission mechanism in the loss function:%
\[
\sigma_{\pi_{CL}}^{2}=\sigma_{\pi_{OL}}^{2}+\sigma_{i}^{2}+2\rho_{i\pi_{OL}%
}\sigma_{i}\sigma_{\pi}<\sigma_{\pi_{OL}}^{2}\Rightarrow-1\leq\rho_{i\pi_{OL}%
}<-\frac{1}{2}\frac{\sigma_{i}}{\sigma_{\pi_{OL}}}.
\]

Friedman (1953) discusses verbally the effect of lags of the transmission
mechanism. His mathematical formulation of the model, however, has no lags. As
mentioned by Phillips (1954b), a negative feedback rule for a static model can
drive instability in a dynamic model.

\subsection{A Conflict between the Kalman Filter and Stochastic Optimal
Control?}

Additionally to the linear quadratic regulator (LQR, Kalman (1960a)) Kalman
developed a filter-estimate that recursively takes in each period the new data
of the current period into account (Kalman (1960b)). This filter is used, for
example, in the global positioning system (GPS) since its inception in 1972.

\begin{quotation}
It is fitting that at the conference [1st world IFAC conference, Moscow, 1960]
Kalman presented a paper, 'On the general theory of control systems' (Kalman,
1960) that clearly showed that a deep and exact duality existed between the
problems of multivariable feedback control and multivariable feedback
filtering and hence ushered in a new treatment of the optimal control problem.
(Bennett, 1996, p.~22).
\end{quotation}

Optimal trajectories derived from the linear quadratic regulator (LQR) and
optimal state estimation using Kalman's filter for linear quadratic estimation
(LQE) were unified in the linear quadratic Gaussian (LQG) system. It is one of
the main tools of Stochastic Optimal Control, see Stengel (1986) and Hansen
and Sargent (2008). The Kalman filter has been quickly implemented in rational
expectations theory, even though an insufficient emphasis on measurement can
be observed.

The transplant of stochastic optimal\ control and the Kalman filter to
macroeconomics was very fast. For example, during his Ph.D. (1973) in
Stanford, Taylor participated in seminars that discussed books linking control
and time-series, such as Whittle's (1963) \textit{Prediction and Regulation
}or Aoki's (1967) \textit{Optimization of stochastic systems}, see Leeson and
Taylor (2012). Hansen and Sargent (2007) confirm:

\begin{quotation}
A profitable decision rule for us has been, `if Peter Whittle wrote it, read
it.' Whittle's 1963 book \emph{Prediction and Regulation by Linear Least
Squares Methods} (reprinted and revised in 1983), taught early users of
rational expectations econometrics, including ourselves, the classical time
series techniques that were perfect for putting the idea of rational
expectations to work. (Hansen and Sargent (2008), p.~xiii)
\end{quotation}

Independently of Kalman (1960b), Muth (1960) updated conditional expectations
based on new information in a simple model. Hansen and Sargent (2007)
highlight that Muth (1960) is a particular case of Kalman's (1960b) filter
estimation. In the early 1970s, following Muth (1960), Lucas solved several
rational expectations models using variants of a Kalman filter (Boumans
(2020)). Most of the time, Lucas used a Kalman filter as a tool for solving
theoretical models of rational expectations, instead of using it for what the
Kalman filter is designed for in the field of control, namely the practical
empirical estimation of the state variables and of the parameters of the
transmission mechanism using time series.

Based on the observations of a few Lucas' papers in the early 1970s and
following Sent (1998), Boumans (2020) infers that engineering mathematics did
not unify optimal trajectories and optimal state estimation, whereas they were
in fact unified in the early 1960s in stochastic optimal control:

\begin{quotation}
Engineering mathematics, however, is not one and unified field. This paper
shows that the mathematical instructions come from informational mathematics,
which should be distinguished from control engineering. (Boumans (2020)).
\end{quotation}

It may appear that the Kalman filter and stochastic optimal control are
distinct approaches because Lucas used a Kalman filter in theoretical papers
in the early 1970s and later dismissed stochastic optimal control for policy
makers in the Lucas (1976) critique. In basic engineering textbooks, such as
Stengel (1986), however, optimal trajectories (Boumans' (2020)
\textquotedblleft control engineering\textquotedblright) and optimal state
estimation including the Kalman filter (Boumans' (2020) \textquotedblleft
informational mathematics\textquotedblright) are merged into stochastic
optimal control:

\begin{quotation}
Not surprisingly, the control principle of chapter 3 and the estimation
principles of chapter 4\ can be used together to solve the stochastic optimal
control problem. (Stengel (1986), p.~420)
\end{quotation}

This artificial split between policy makers' optimal trajectories and optimal
state estimation done by Lucas was in practice driven by the rules
\emph{versus} discretion controversy. This explains Lucas' change of
perspective and tools. Optimal control, and especially stochastic optimal
control with a policy maker reestimating in each period the optimal policy is
clearly a method favored by macroeconomists in the camp \textquotedblleft
discretion\textquotedblright. Hence, Lucas, as a proponent of
\textquotedblleft rules\textquotedblright, took the opportunity to break the
dual approaches unified by stochastic optimal control in two parts. The
advocacy of particular theoretical views taking sides in a controversy and the
selection of mathematical tools from control goes hand-in-hand. Scientific
tools are part of the scientific rhetorics in order to convince a majority of
scientists to opt for one side of the controversy.

\begin{quotation}
The present study, then, provides a rationalization for rules with smooth
monetary policy, exactly as did the earlier studies of Lucas, Sargent and
Wallace, and Barro. Similarly, it rationalizes the analogous fiscal rule of
continuous budget balancing and rules to stabilize the quantity of private
money, such as larger reserve requirements for banks. (Lucas (1975), p.~1115)
\end{quotation}

Barnett (2017) identifies the different emphasis on measurement as the main
methodological difference between rocket science and macroeconomics, with both
fields using control methods. William A. Barnett had a BS degree in mechanical
engineering from MIT in 1963. He worked as an \textquotedblleft rocket
scientist\textquotedblright\ at Rocketdyne from 1963 to 1969, a contractor for
the Apollo program. He then got a Ph.D. in statistics and economics from
Carnegie Mellon University in 1974 at the same time as Finn Kydland. He worked
for the models of the Fed soon after.

\begin{quotation}
The different emphasis on measurement is very major, especially between
macroeconomics and rocket science... In real rocket science, engineers are
fully aware of the implications of systems theory, which emphasizes that small
changes in data or parameters can cause major changes in system dynamics. The
cause is crossing a bifurcation boundary in parameter space... But when policy
simulations of macroeconometric models are run, they typically are run with
the parameters set only at their point estimates. For example, when I was on
the staff of the Federal Reserve Board, I never saw such policy simulations
delivered to the Governors or to the Open Market Committee with parameters set
at any points in the parameter estimators' confidence region, other than at
the point estimate. This mind-set suggests to macroeconomists that small
errors in data or in parameter estimates need not be major concerns, and hence
emphasis on investment in measurement in macroeconomics is not at all
comparable to investment in measurement in real rocket science. (Barnett
(2017), p.22)
\end{quotation}

Lucas, Kydland and Prescott did not increase the emphasis on measurement with
respect to modelers at the Fed in the 1970s. They had prior theoretical views
taking sides for ``rules'' in the debate on rules \emph{versus} discretion. As
well as optimal control, using a Kalman filter for optimal state estimation
using macroeconomic time series would give a chance to challenge these prior
views. Hence, it was a rhetorical strategy in the rules \emph{versus}
discretion controversy to challenge measurement and identification strategies.

Lucas (1976) presented as impossible the econometric identification of
parameters in the case of reverse causality due to feedback rules. When
confronting their model with data, Kydland and Prescott (1982) decided not to
use Kalman-filter optimal-state estimations and avoid identification issues,
but did simulations. The increased distance between econometricians and
new-classical macroeconomists with respect to identification issues in
econometric estimations using a Kalman filter, led to a selection bias with
respect to the tools and the practice of control. In the field of control, the
measurement of the transmission mechanism is viewed as much more important
than theoretical models based on a priori hypothesis taking sides in a
controversy. This increased the gap between the macroeconomics of
stabilization policy and the field of control in the 1980s.

Kalman, who worked for the Apollo program, emphasized that the most important
issue for stabilization of dynamic systems is the \emph{accuracy of the
measurement} and the estimation of parameters of the system of policy
transmission mechanism ($A$ and most importantly $B$ and its sign) instead of
a priori modeling of the real world. For Kalman, measurement is above all
necessary for optimal policy rules, so one should not ``separate'' filtering
estimation with data from optimal trajectories found with control. One should
not separate macroeconomic theory from econometrics using macroeconomic
time-series. By contrast, the linear quadratic gaussian (LQG) system merging
optimal policy and Kalman filter effective estimation with macroeconomic time
series has not been widely used during the the 1980s.

\subsection{The Delayed Transplant of Robust Control}

The third stage of control theory, namely robust control (Zhou et al. (1996),
Hansen and Sargent (2008)), emerged in the following of Doyle's (1978)
two-pages paper presenting a counter-example where the Kalman linear quadratic
Gaussian (LQG) model has not enough guaranteed margins to ensure stability.

Robust control assumes that the knowledge of the transmission parameters is
uncertain on a known \emph{finite interval} with a given sign. In our example,
it may correspond to $B_{\min}<B<B_{\max}<0$. An evil agent tries to fool as
much as possible the policy maker, which reminds of Descartes (1641) evil demon:

\begin{quotation}
I will suppose therefore that not God, who is supremely good and the source of
truth, but rather some malicious demon, had employed his whole energies in
deceiving me. (Descartes (1641))
\end{quotation}

The policy maker is the leader of a Stackelberg dynamic game against the evil
agent. The policy maker is minimizing the maximum of the losses in the range
of uncertainty on the parameters. Kalman's solutions of matrix Riccati
equation for the linear quadratic regulator remain instrumental for finding
the solutions of robust control.

For monetary policy, the uncertainty on parameters was put forward by Brainard
(1967), so that, indeed, this is a very important issue. But the transplant of
the new methods of robust control was delayed in the 1980s. Very soon after
robust control tools emerged, Von\ Zur Muehlen (1982) wrote a working paper
applying robust control to monetary policy which was never published in an
academic journal. No major conference took place between leaders in the field
of robust control and renowned macroeconomists like the ones described by
Kendrick (1976) for stochastic optimal control in the early 1970s. Instead of
at most three years for the spread of new tools from the field of control, one
had to wait until the end of the 1990s, more than fifteen years. For now two
decades, Hansen and Sargent (2008, 2011) transplanted robust control tools
into macroeconomics (Hansen and Sargent (2008)), with still a limited number
of followers in macroeconomics:

\begin{quotation}
When we became aware of Whittle's 1990 book, \textquotedblleft\textit{Risk
Sensitive Control\textquotedblright} and later his 1996 book \textquotedblleft%
\emph{Optimal Control: Basics and Beyond\textquotedblright}, we eagerly worked
our ways through them. These and other books on robust control theory like
Basar and Bernhard's $H^{\infty}$\emph{ \textquotedblleft Optimal Control and
Related Minimax Design Problems: A Dynamic Game\textquotedblright\ Approach}
provide tools for rigorously treating the `sloppy' subject of how to make
decisions when one does not fully trust a model and open the possibility of
rigorously analyzing how wise agents should cope with fear of
misspecification.' (Hansen and Sargent (2008), p.~xiii)
\end{quotation}

The explanation of the delay for the transplant of robust optimal control is
that the policy ineffectiveness arguments spread and took over for a while
among influential macroeconomists, according to Kendrick (2005):

\begin{quotation}
However, it was a problem in the minds of many! As a result, work on control
theory models in general and stochastic control models in particular went into
rapid decline and remained that way for a substantial time. In my judgment, it
was a terrible case of `throwing the baby out with the bath water'. The work
on uncertainty (other than additive noise terms) in macroeconomic policy
mostly stopped and then slowly was replaced with methods of solving models
with rational expectations and with game theory approaches.... I believe that
the jury is still out on the strength of these effects and think that there
was a substantial over-reaction by economists when these ideas first became
popular. (Kendrick (2005), p.15)
\end{quotation}

We now re-evaluate the Lucas critique and the time-inconsistency theory that
are at the origin of the long delay for using robust control methods and the
relative scarcity of applied macroeconometric estimations using a Kalman
filter jointly with optimal policy with the linear quadratic Gaussian model.

\section{The Impossibility to Identify the Parameters of the Policy
Transmission Mechanism}

\subsection{Lucas (1976) Critique}

In popular terms, the Lucas critique states that (Keynesian) macroeconomic
relations change with government policy. The Lucas critique is a parameter
identification problem in a system of dynamic equations including reverse
causality due to a feedback control rule equation. The model of the last
section of the paper of the Lucas critique (Lucas, 1976) can be stated in the
form of the first-order single-input single-output model. We add our notations
of the previous sections in brackets in the Lucas (1976) quotation. The
corresponding notations are: for the policy target: $y_{t}=\pi_{t}$, for the
policy instrument: $x_{t}=i_{t}$, and for the parameters of the transmission
mechanism: $\theta=\left(  A,B\right)  $ which also includes additive random
disturbances $\varepsilon_{t}$. In order to avoid confusion with the control
notation $F$ for the policy rule parameter, we change Lucas' notation of
function $F(.)$ into $D(.)$:

\begin{quotation}
I have argued in general and by example that there are compelling empirical
and theoretical reasons for believing that a structure of the form%
\[
y_{t+1}=D\left(  y_{t},x_{t},\theta,\varepsilon_{t}\right)  \text{ [our
notation: }\pi_{t+1}=A\pi_{t}+Bi_{t}+\varepsilon_{t}\text{] }%
\]
$D\left(  .\right)  $ known, $\theta$\ fixed, $x_{t}$ arbitrary will not be of
use for forecasting and policy evaluation in actual economies... One cannot
meaningfully discuss optimal decisions of agents under arbitrary sequences
$\left\{  x_{t}\right\}  $ of future shocks. As an alternative
characterization, then, let policies and other disturbances be viewed as
stochastically disturbed functions of the state of the system, or
(parametrically)%
\[
x_{t}=G\left(  y_{t},\lambda,\eta_{t}\right)  \text{ [our notation: }%
i_{t}=F\pi_{t}+\eta_{t}\text{]}%
\]
where $G$ is known, $\lambda$ is a fixed parameter vector, and $\eta_{t}$ a
vector of disturbances. Then the remainder of the economy follows%
\[
y_{t+1}=H\left(  y_{t},x_{t},\theta\left(  \lambda\right)  ,\varepsilon
_{t}\right)  \text{ [or }\pi_{t+1}=\left(  A+BF\right)  \pi_{t}+\varepsilon
_{t}+B\eta_{t}\text{]}%
\]
where, as indicated, the behavioral parameters $\theta$ [or $A+BF$] vary
systematically with the parameters $\lambda$ [or $F$] governing policy and
other \textquotedblleft shocks\textquotedblright. \textbf{The econometric
problem in this context is that of estimating the function} $\theta\left(
\lambda\right)  $ [or $A+BF$]. In a model of this sort, a policy is viewed as
a change in the parameters $\lambda$ [or $F$] or in the function generating
the values of policy variables at particular times. A change in policy (in
$\lambda$ [or $F$]) affects the behavior of the system in two ways: first by
altering the time series behavior of $x_{t}$; second by leading to
modification of the behavioral parameters $\theta\left(  \lambda\right)  $ [or
$A+BF$] governing the rest of the system. (Lucas (1976, p.39-40))
\end{quotation}

The closed loop parameter $\left(  A+BF\right)  $ determines the persistence
of policy targets. It depends on the policy rule parameter $F$ as in all
closed loop models including feedback mechanism in the field of control. The
econometric problem is the identification of parameters $A$, $B$ and $F$ in a
system of dynamic equations with reverse causality. In addition, in the scalar
case, the true parameters $B$ and $F$ have opposite signs for the case of a
negative-feedback mechanism ($BF<0$).

Lucas (1976) over-stated that this parameter identification problem cannot
\emph{in principle} be addressed.

\begin{quotation}
The point is rather that this possibility [of feedback rules] cannot
\textit{in principle} be substantiated empirically. (Lucas (1976, p.41))
\end{quotation}

Because the identification of estimates of $A$ and $B$ cannot in principle be
substantiated empirically, then it is better to use rules such as $i_{t}=0$
and $F=0$ instead of \textquotedblleft discretion\textquotedblright\ with a
feedback rule parameter $F\neq0$.

This reverse-causality parameter identification problem is identical to the
one of the private sector's demand and supply price elasticities with respect
to quantities, when there is negative feedback of supply on demand (section
2.1, Smith (1776) and Ezekiel's (1938) cobweb model). The supply of goods
should increase when price increases, whereas the demand of goods should
decrease when price increases. But there is \emph{only one covariance} for
observed prices and quantities with a \emph{single sign}.

According to Koopmans (1950), the identification of $B$ and $F$ finding
opposite signs can \emph{in principle} be substantiated empirically, if one
finds strong and exogenous instrumental variables with distinct identification
restrictions for both the transmission mechanism (demand) and the feedback
equation (supply). Finding these instrumental variables with their
identification restrictions is far from being easy with macroeconomic data.
One answer is instrumental variables allowing identification via
identification restrictions such as exogenous monetary policy shocks,
exogenous inflation shocks or lags of some explanatory variables in vector
autoregressive (VAR) models which include policy-target (inflation) and
policy-instrument (funds rate) equations or lags of explanatory variables.

Although Lucas (1976) does not mention microeconomic foundations, it is often
erroneously claimed that microeconomic foundations are required because of the
Lucas (1976) critique. From private sector's demand and reverse causality of
supply with negative feedback, the private sector's behavior estimates of
elasticities are also subject to the Lucas (1976) critique. If the private
sector's behavior is modeled with microfoundations (section 3.2, intertemporal
optimization using optimal control), the representative household's behavior
includes an optimal policy-rule equation $F^{\prime}$ besides the law of
motion of the private sector's state variables with parameters $(A^{\prime
},B^{\prime})$ as in section 2.3. The Lucas critique applies as well: there is
a system of two equations with reverse causality and opposite signs for these
two parameters: $B^{\prime}F^{\prime}<0$. The private sector persistence
parameter $A^{\prime}+B^{\prime}F^{\prime}$ is not a structural parameter.
Therefore, private sector's micro-foundations are never an \textquotedblleft
answer\textquotedblright\ to the parameter identification problem of the Lucas
(1976) critique. Dynamic models of the private sector alone face also the
Lucas critique, even without policy intervention where $BF=0$ using section 2.3\ notations.

Lucas' rhetorical over-statement was successful, according to Blinder's
interview in July 1982 reported in Klamer (1984):

\begin{quotation}
This is a case where people have latched upon a criticism. All you have to do
in this country (more than in other places) right now is scream mindlessly,
\textquotedblleft Lucas critique!\textquotedblright\ and the conversation
ends. That is a terrible attitude. (Klamer (1984), p.166)
\end{quotation}

Instead of fostering a research agenda on identification issues with reverse
causality in applied econometrics, it convinced a number of macroeconomists
that applied macroeconometrics was to be dismissed and hopeless. Macroeconomic
theorists favored theory with the private sector's microeconomic foundations
using simulations, instead of tackling thorny identification issues in econometrics.

\subsection{Sims (1980) faces the Lucas critique: The VAR\ Price Puzzle for
Identifying $B$ and $F$ with Opposite Signs}

Several Keynesian applied econometricians, including Blinder, investigated the
stability of reduced form estimates $A+BF$ without finding much change on
several key equations from 1960 until 1979 or 1982 (Goutsmedt et al. (2019)).
But it is possible to have unchanged reduced form persistence parameters, even
when the \textquotedblleft structural\textquotedblright\ parameters did change
from period 1\ (before 1973) to period 2 (after 1973): $A_{1}+B_{1}F_{1}%
=A_{2}+B_{2}F_{2}$. The difficult step is to identify separately $A$, $B$ and
$F$ for each period, knowing that there is a reverse causality with opposite
signs ($BF<0$) of the policy instrument and the policy target in the
transmission mechanism (parameter $B$), on the one hand, and in the policy
rule (parameter $F$), on the other hand.

The Lucas critique was used as an explanation for the forecasting errors of
large scale models such as the FMP model, for Federal Reserve-MIT-Penn.
Another critique of large scale forecasting models was put forward by Sims
(1980) with his vector auto-regressive models (VAR). Sims (1980) warned that
some identification restrictions (for example, constraining some parameters to
be zero) were not tested in macroeconometrics in the 1970s.

But soon an unexpected anomaly of VAR\ models appeared, labeled the
\textquotedblleft price puzzle\textquotedblright\ (Rusn\'{a}k et al. (2013)).
We highlight that the VAR\ price puzzle is closely related to the Lucas (1976)
critique. Impulse response functions obtain that inflation increases,
following a shock increasing the funds rate. This is the opposite estimate of
the expected sign of the true model: $\widehat{B}>0>B$. Indeed, impulse
response functions are such that the funds rate increases, following a shock
increasing inflation, with the expected sign: $\widehat{F}>0$. The identical
sign for $\widehat{B}$ and $\widehat{F}$ comes from the fact that the sign of
the covariance between the policy instrument and the policy target does not
change with one or two lags. The price puzzle is such that $\widehat{B}>0$ and
$\widehat{F}>0$, so that there seems to be a positive feedback mechanism on
inflation persistence of policy $\widehat{A}+\widehat{B}\widehat{F}>A$, even
during Volcker's disinflationary policy with negative-feedback mechanism.

Estimating the correct sign of $B$ is of utmost importance for policy advice.
Negative feedback reduces the persistence of the policy target according to
the condition: $0<A+BF<A$. This implies $BF<0$. The sign of $F$\ should be the
opposite of the sign of $B$. If one estimates wrongly that $\widehat{B}>0$
whereas its true value is $B<0$, the policy advice will be $F<0$, so that
$\widehat{B}F<0<BF$. Therefore, the policy advice will wrongly increase the
persistence of the policy target:
\begin{equation}
0<A+\widehat{B}F<A<A+BF\text{ if wrong estimated sign }B<0<\widehat{B}.
\end{equation}

Several econometric techniques have been used to reverse the sign of the
\textquotedblleft price puzzle\textquotedblright\ (Rusn\'{a}k et al. (2013)).
Perhaps some of them correctly managed to answer the Lucas (1976) critique for
some data set. Walsh's (2017) textbook restricts the description of the price
puzzle measurement issue to half a page. Walsh's (2017) textbook demonstrates
that it has been easier to develop\ a large variety of theoretical
micro-founded macroeconomic models conflicting among themselves than to
provide accurate measurements of $A$, $B$ and $F$.

\subsection{Kydland and Prescott (1982) face the Lucas' critique}

For a representative agent of the private sector, Ramsey (1928) models a
negative-feedback mechanism for optimal savings or optimal consumption $x_{t}$
with consumption smoothing (related to the notations $R/Q$ for the linear
quadratic approximation). This idea was formalized again independently by
Cass, Koopmans and Malinvaud in 1965\ (Spear and Young (2014)). Kydland and
Prescott (1982) stuck to this idea adding autocorrelated productivity shocks
$z_{t}$.

Their model is a stock-flow reservoir model where the stock of wealth or the
stock of capital $k_{t}$ has an optimal set point $k^{\ast}$ (De Rosnay
(1979), Meadows (2008)). The stock of capital is controlled by
saving/investment inflows and it faces depreciation outflows. If the level of
the stock is below its long run optimal target, consumption decreases
proportionally according to a negative feedback proportional rule:
\[
x_{t}=F^{\prime}k_{t}+F_{z}^{^{\prime}}z_{t}\text{, }k_{t+1}=\left(
A^{\prime}+B^{\prime}F^{\prime}\right)  k_{t}+\left(  A_{kz}^{\prime
}+B^{\prime}F_{z}^{\prime}\right)  z_{t}\text{ with }0<A=A^{\prime}+B^{\prime
}F^{\prime}<1\text{ ,}%
\]

where both the consumption flow and the capital stock are written in deviation
from their long run target. Consumption also responds to an exogenous
auto-correlated productivity shock $z_{t}$. To become rich, a poor
representative households with capital below its long run target ($k_{t}<0$)
\emph{only needs to save more} which implies his optimal consumption is below
its long run target ($x_{t}<0$). Hence, the poor representative household
replenishes his reservoir storing his stock of wealth to its optimal level
($k_{t}=0$).

When the relative cost of changing consumption ($R/Q$) increases, the policy
rule parameter decreases (consumption is less volatile) and the persistence of
capital or wealth increases $A$ or equivalently, the speed of convergence of
capital towards equilibrium decreases.

Kydland and Prescott (1982) send us back to Smith's (1776) view that a
negative-feedback mechanism always works within the private sector with supply
and demand interaction. This time, the negative-feedback mechanism is related
to the investment/saving market. It is a typical model including microeconomic
foundations which faces the Lucas critique\footnote{Sergi (2018) provides the
history of the Lucas (1976) critique in relation to real business cycles
models and dynamic stochastic general equilibrium models which followed
Kydland and Prescott (1982) paper.}:

(1) There exists a problem of parameter identification due to a private sector
reverse causality feedback rule. In particular, it is difficult to find
properly identified opposite signs for optimal behavior of the private sector
with negative feedback such that: $B^{\prime}F^{\prime}<0$.

(2) They assume and do not test the identification restrictions such that
$BF=0$. They assume that there is no transmission mechanism of macroeconomic
policy ($B=0$) and no policy maker's feedback rule ($F=0$). This amounts to
consider that the autocorrelation of the policy target $A+BF$ is exogenous.
But, if one does not test whether $BF=0$ or not, one cannot exclude $BF\neq0$
so that the persistence of the policy target is a reduced form parameter which
depends on the policy rule parameter $F$ according to $A+BF$. This is exactly
the Lucas (1976) critique, as highlighted by Ingram and Leeper (1990):

\begin{quotation}
Kydland and Prescott's model assumes that policy doesn't affect private
decision rules. There is no policy evaluation to perform. Alternatively, if
policy does affect private behavior, then the parameters Kydland and Prescott
calibrate are reduced-form parameters for some underlying model embedding
monetary and fiscal policy. Thus, if there is any policy evaluation left to
perform, Kydland and Prescott's calibrated parameters must be functions of
policy behavior and should change systematically with policy. (Ingram and
Leeper (1990), p.~3)
\end{quotation}

Kydland and Prescott (1982) have chosen simulations and calibrations in order
to avoid the econometric measurement of $A^{\prime}$, $B^{\prime}$,
$F^{\prime}$, $B$, and $F$. To solve and simulate their model, Kydland and
Prescott (1982) used Kalman (1960a) LQR\ solutions of a Riccati equation.
Without Kalman's (1960a) solution of optimal control for the private sector,
real business cycles would not have been computed. Lucas selected a Kalman
filter for his theoretical papers, but rejected optimal control for policy
makers doing \textquotedblleft discretion\textquotedblright, Lucas (1976).
Kydland and Prescott (1982) selected optimal control for the private sector
only, but rejected Kalman filter estimations. Because Kydland and Prescott
(1982) expected their model to be rejected by econometrics, they decided to
avoid a measurement strategy based on a Kalman filter, as opposed to
researchers in the field of control. Behind this seemingly inconsistent use of
tools of control by Lucas and Kydland and Prescott, there remains a common
motive: support \textquotedblleft rules\textquotedblright\ in the rules
\emph{versus} discretion controversy. Write complicated and technical papers,
imported from the tools of control, with a crucial hypothesis (here $BF=0$)
which handicaps policy maker's negative-feedback discretion with respect to rules.

\begin{quotation}
Only when we have considerable confidence in a theory of business cycle
fluctuations would the application of public finance theory to the question of
stabilization be warranted. Such an extension is straightforward in theory,
though in all likelihood carrying it out will be difficult and will require
ingenuity. (Kydland and Prescott (1988), p.~358)
\end{quotation}

In relation to Kydland and Prescott's (1982) restriction $BF=0$ of no
stabilization policy during post-war business cycles, Blinder states in an
interview in July 1982:

\begin{quotation}
One important aspect to the debate is how stable an economy would be without
an active government... Look at a time series chart (in my stagflation book)
of the year-to-year fluctuations of GNP. They are remarkably smaller in the
'50s and '60s, the \textquotedblleft Keynesian\textquotedblright\ years, than
they were before... I claim I know why: that there was active fiscal and
monetary policy to iron business cycles beginning after World War II, but not
before. So I think discretionary policies have been used and they have often
worked. (Klamer (1984), p.~165)
\end{quotation}

Lucas (1987) confirms Blinder's observations for the case of consumption time
series, which has in general lower volatility then GNP:

\begin{quotation}
In the period prior to the Second World War, and extending as far back in time
that we have usable data, the standard deviation (logarithmic deviations from
trend) of consumption was about three times its post-war level. (Lucas (1987), p.~28)
\end{quotation}

Although Simons' (1936) and Friedman's (1948) 100\% reserve requirement for
banks was not enacted, the '50s and '60s were decades without banking crisis
following the Glass-Steagall act of 1933, separation of investment and retail
banking and the Bretton Woods conference, 1944, with capital controls and
international financial stability. Financial regulation and high productivity
growth are likely to have been complementary to active stabilization policy in
order to damp business cycles fluctuations during these two decades.

\subsection{Lucas (1987) faces the Lucas critique}

Lucas (1987, chapter\ 3 and 2003) evaluates the cost of aggregate US
consumption fluctuations by its certainty equivalent loss of consumption. He
uses the second order Taylor development of a discounted separable utility
function with constant relative risk aversion of the representative household.
This certainty equivalent is one half of the coefficient of absolute
fluctuations aversion (or relative risk aversion) $\gamma$ times the variance
of the log of U.S. aggregate consumption around its trend denoted $\sigma
_{x}^{2}$. Lucas (1987) considers the following range of estimates
$\widehat{\gamma}$ for the coefficient of constraint relative risk aversion:%

\[
\frac{1}{2}\gamma\sigma_{x}^{2}\text{ with }1\leq\widehat{\gamma}<20\text{ and
}\widehat{\sigma}_{x}= 0.013.
\]

Lucas (1987) is the transcript of Yrj\"{o} Jahnsson lecture in May 1985 and
Lucas (2003) is the transcript of January 2003 presidential address at the
American Economic Association, eighteen years after for the same evaluation.
Lucas (1987) evaluates the certainty equivalent cost of post Second World War
fluctuations of consumption to $\$8.50$ year 1983 per person ($\$22$ year
2019). This is an extremely low cost. But it is also a conservative estimate
of welfare because Lucas' (1987) utility does not include the leisure benefits
during recessions due to reduction of working hours which partially offsets
the loss of consumption in Kydland and Prescott's (1982) utility. Lucas (1987,
p.28) mentions \textquotedblleft I would guess that taking leisure
fluctuations into account more carefully would reduce the estimate in the text
still further.\textquotedblright

Lucas' (1987) low estimate of the variance of detrended consumption
$\sigma_{x}^{2}$ since the Second World War faces the Lucas critique for the
following reason: It is a reduced-form estimate $\sigma_{x}^{2}$\emph{ which
depends on stabilization policy parameters }$F$. For example, if the log of
U.S. aggregate consumption around its trend is a stationary autoregressive
process of order one, its variance decreases if there is negative feedback due
to stabilization policy.%

\begin{align*}
x_{t}  &  =\left(  A+BF\right)  x_{t-1}+\varepsilon_{t}\text{ and }%
0<A+BF<\min(A,1)\\
\sigma_{x}^{2}  &  =\frac{\sigma_{\varepsilon}^{2}}{1-\left(  A+BF\right)
^{2}}<\frac{\sigma_{\varepsilon}^{2}}{1-A^{2}}\text{ if }0<A<1\text{ and
}BF<0.
\end{align*}

Lucas (1987) notes the difference of the volatility of consumption before and
after Second World War:

\begin{quotation}
Fluctuations in the pre-Second World War, especially combined as they were
with an absence of adequate programs for social insurance, were associated
with large cost of welfare... In the period prior to the Second World War, and
extending as far back in time that we have usable data, the standard deviation
(logarithmic deviations from trend) of consumption was about three times its
post-war level. Since this number is squared in the formula (8), the implied
cost estimates are multiplied by nine, becoming like one-half of 1\ per cent
of total consumption. As deadweight losses go, \emph{this is} a large number.
(Lucas (1987), p.28)
\end{quotation}

But he does not attribute the change of the volatility of consumption
$\sigma_{x}^{2}$ to the changes of the policy-rule parameter $F$ related to
the introduction of Keynesian stabilization policy. His conclusion therefore
may face the Lucas (1976) critique:

\begin{quotation}
I find the exercise instructive, for it indicates that economic instability at
the level we have experienced since the Second World War is a minor problem,
even relative to historically experienced inflation and certainly relative to
the costs of modestly reduced rates of economic growth (Lucas (1987), p.~31).
\end{quotation}

Blinder's argument (interview in July 1982 in Klamer (1984)) is to compare the
US\ economy after the Second World War with stabilization policy ($BF<0$) to
the US\ economy before this period without stabilization policy (with e.g.
$BF=0$). One would have saved up to $8$ times the cost of consumption
fluctuations using Lucas (1987) utility and estimates: $\$68$ year
1983\ ($\$175$\ year 2019)\ per person, still a small benefit. Then, Blinder
(1987) changes the scale of Lucas evaluation of the cost of post-war-II
US\ business cycles as follows:

\begin{quotation}
Now change the utility to the Stone-Geary form: $U=log(C-\$1500)$. Here, a
4\ percent drop in consumption reduces utility by 8.3 percent... Finally, let
the cycle instead reduce the consumption of 10 percent of the population by 40
percent while the other 90 percent loses nothing. (Note I am allowing very
generous unemployment insurence here.) With the Stone-Geary utility function,
mean utility declines 16.1\ percent. (Klamer (1984), p.~165)
\end{quotation}

Since then, the great financial crisis in 2007\ and its jobless recovery
changed the figures of the standard error of US\ detrended aggregate
consumption since the Second World War. Twenty-two years after Lucas Yrj\"{o}
Jahnsson lecture in 1985, Barro (2009) uses Epstein-Zin-Weill utility for a
representative agent and estimates that economic disasters with at least 15\%
loss of GDP\ over several consecutive years has a frequency of 2\% for
35\ countries during the 20th century:

\begin{quotation}
Society would willingly reduce GDP by around 20 percent each year to eliminate
rare disasters. The welfare cost from usual economic fluctuations is much
smaller, though still important, corresponding to lowering GDP by about 1.5
percent each year. (Barro (2009))
\end{quotation}

\subsection{Kydland and Prescott (1977, section 5): Systematic Bias of
Estimated Parameters}

Following the Lucas (1976) critique, because the identification of estimates
of $A$ and $B$ cannot in principle be substantiated empirically, then it is
better to use \textquotedblleft rules\textquotedblright\ where $i_{t}=0$ and
$F=0$ instead of \textquotedblleft discretion\textquotedblright\ with a
feedback rule parameter $F\neq0$. Accordingly, Kydland and Prescott (1977) set
a new definition of discretion where the policy maker always measures the
parameters of the transmission mechanism with cumulated systematic errors
which finally leads to the instability of the policy target.

\begin{definition}
A\ discretionary policy is a policy for which the policy rule is optimal under
the incorrect assumption that the observed persistence of the policy target is
invariant to the policy-rule parameter used in the previous period.
\end{definition}

\begin{quotation}
The policymaker uses control theory to determine which policy rule is optimal
\textbf{under the incorrect assumption} that the equilibrium investment
function is invariant to the policy rule used... Econometricians revise their
estimate of the investment function, arguing that there has been structural
change, and the policymaker uses optimal control to determine a new policy
rule. (Kydland and Prescott (1977), section 5).
\end{quotation}

To describe the private sector, a model with convex adjustment costs of
investment is used. Because of the optimal behavior of the private sector,
\textquotedblleft rules\textquotedblright\ pegging the policy instrument to
its long-run value (where $i_{t}=0$ and $F=0$) has stable dynamics
($0<A=A^{\prime}+B^{\prime}F^{\prime}<1$). For optimal taxation, the policy
maker's instrument is investment tax credit. Because $A\neq0$ and $Q\neq0$,
\textquotedblleft rules\textquotedblright\ are sub-optimal with respect to
optimal control discretion by the policy maker obtaining unbiased estimates of
the policy rule parameter $A$ (section 3.2). Assuming an increasing bias of
the policy-rule parameter $A$, Kydland and Prescott's (1977) simulations lead
to the increased volatility of the policy target in the discretion case when
compared to the rules case:

\begin{quotation}
For one example (...), after the third iteration, however, performance
deteriorated, and the consistent [optimal control] policy to which the process
converged was decidedly inferior to the passive policy for which the
investment tax credit was not varied... For another example, the iterative
process did not converge. Changes in the policy rule induced ever larger
changes in the investment function. The variables fluctuated about their
targeted values but fluctuated with increased amplitude with each
iteration.... Such behavior either results in consistent but suboptimal
planning or in economic instability. (Kydland and Prescott (1977), p.~485)
\end{quotation}

We translate their argument using the first order single-input single-output
model. In each period, the policy maker uses $A+BF_{t-1}$ instead of $A$ for
determining his policy rule parameter $F_{t}$ at the current period. The
policy makers sets an optimal policy $F_{t}$ which belongs to the stability
set for the erroneous $A^{\prime}=A+BF_{t-1}$. This implies that the sequence
of policy-rule parameters $F_{t}$ is strictly increasing or decreasing:%

\[
\left\vert \left(  A+BF_{t-1}\right)  +BF_{t}\right\vert <A<1\text{
}\Rightarrow\left\vert B\right\vert \left\vert F_{t-1}-F_{t}\right\vert <0.
\]

After some iterations up to a date $t_{1}$, the \textquotedblleft
discretionary\textquotedblright\ policy rule $F_{t}$ may belong to a set with
a larger persistence $\left\vert A+BF_{t}\right\vert $ and volatility of the
policy target than the one obtained with \textquotedblleft
rules\textquotedblright\ (where $i_{t}=0$ and $F=0$) persistence equal to
$A$:
\[
\left\vert A+BF_{t}\right\vert >A\text{, for }t>t_{1}.
\]

Because the policy maker estimates with an increasing bias the open-loop
persistence $A$\ of the policy target, discretion implies higher volatility of
the policy target than rules. This result does not assume rational
expectations with forward-looking variables, nor time-inconsistency in a
Stackelberg dynamic game, which is defined in the next section.

\section{The Impossibility of Policy Maker's Credibility and therefore of
Optimal Control.}

\subsection{Simaan and Cruz' (1973b) Time-Inconsistency of Stackelberg Dynamic
Games}

Kydland and Prescott's (1977) expressed their dislike of discretionary policy
in the abstract of their famous paper on rules \textit{versus} discretion:

\begin{quotation}
... economic planning is not a game against nature but, rather, a game against
rational economic agents. We conclude that there is \emph{no} way control
theory can be made applicable to economic planning when expectations are
rational. (Kydland and Prescott (1977), p.~473)
\end{quotation}

This radical statement put forward the specificity and the autonomy of
macroeconomists with respect to the field of control. Following this statement
of such a divide as a major scientific contribution, surveys and
historiography of the time-inconsistency literature done by economists do not
mention Simaan and Cruz's (1973b) original contribution, for example Hartley
(2006). Even Tabellini (2005) mentions Simaan and Cruz only in a footnote.

The point-of-view from the field of control is markedly different. The field
of control was subsidized by the army in the 1940s and during the cold war.
War is related to non-cooperative dynamic games, so army funding was not
limited to games against nature. Dynamic Nash games (Isaacs (1965)) used tools
related to the field of control. Dynamic games flourished in engineering and
applied mathematics departments using control tools. It was marginal in
economics departments. Cruz, a professor of electrical and computer
engineering at university of Illinois, published with his student, Simaan, an
important result on Stackelberg dynamic games (Simaan and Cruz (1973b)). The
paper was partly funded by the US Air Force. Cruz (2019) describes the
following sequence, from the point of view of the field of control:

\begin{quotation}
R.P. Isaacs (Isaacs 1965) is the originator of differential games. Y.C. Ho (Ho
1970) clarified the connection between control and differential game theory.
C.I. Chen and J.B. Cruz (Chen and Cruz 1972) were the first to consider
dynamic Stackelberg games [with an economic example]. M.A. Simaan and J.B.
Cruz (Simaan and Cruz 1973a, Simaan and Cruz 1973b) reframed dynamic
Stackelberg game theory, providing mathematical proofs for various results and
showed that the Stackelberg strategies do not necessarily satisfy Bellman's
principle of optimality. This violation of the principle of optimality was
renamed as time inconsistency in economic policy (Kydland and Prescott 1977),
and later Kydland and Prescott won the Nobel prize in economics in 2004 in
part on the basis of their paper. M.A. Simaan and J.B. Cruz (Simaan and Cruz
1973c) extended the dynamic Stackelberg concept to many players. J.B. Cruz
(Cruz 1975) introduced the dynamic Stackelberg concept to the economics
community, and he extended the Stackelberg strategy to more than two levels
(Cruz 1978). (Cruz (2019), p.83)
\end{quotation}

Simaan and Cruz (1973b) is one of the key building blocks for Kydland and
Prescott's (1977) later success. Theoretical work done by two electrical
engineers was the initial trigger of Kydland and Prescott's (1977) and
Prescott's (1977) papers denying the use of control for stabilization policy.

Stackelberg dynamic games can consider a quadratic loss function for the
leader and the follower. In this case, their solutions are extensions of
Kalman's linear quadratic regulator solutions, solving an algebraic Riccati equations.

The leader takes into account the follower's marginal conditions. However,
like the tower of control interacting with the pilot of an airplane, the
leader can order the change of the decision variables of the follower (the
pilot in the cockpit). These decision variables are jump variables. Hence, the
policy maker can anchor the initial value of these jump variables optimally,
minimizing the loss function with respect to the follower's decision variable
at the initial date.

In our example, let us assume that prices and inflation are private sector's
\textquotedblleft jump\textquotedblright\ decision variable, with unknown
initial condition, instead of assuming a given initial value $\pi_{0}$. This
initial transversality condition is optimally decided by the policy maker,
where the anchor parameter $P^{\ast}\left(  Q/R\right)  $ of the jump variable
on predetermined variables $z_{0}$ is derived from optimization. It depends on
the policy maker's preferences:%

\begin{equation}
\gamma_{0}^{\ast0}=\left(  \frac{\partial L}{\partial\pi}\right)
_{0}=0\Rightarrow\pi_{0}^{\ast0}=P^{\ast}(Q/R)z_{0}\text{ and }\gamma
_{t}^{\ast0}>0=\gamma_{0}^{\ast0}\text{ for }t=1,..
\end{equation}

The time superscript besides a star $\ast0$ indicates it is a date-$0$ optimal
plan. This plan is valid for a list of future dates starting at date $t=0$
which corresponds to the time-subscript. Simaan and Cruz's (1973b) result is
that if the policy maker re-optimizes at any future period indexed by $t$, it
contradicts the optimal plan decided at date zero: $\gamma_{t}^{\ast
t}=0<\gamma_{t}^{\ast0}$. They conclude:

\begin{quotation}
If the starting time is fixed, the leader's closed-loop Stackelberg control is
the best control law (among all other admissible closed-loop controls) that he
can announce prior to the start of the game, but it does not have this same
desirable property from any other starting time (Simaan and Cruz's (1973b), p.~625).
\end{quotation}

In a personal communication, Simaan recalls the genesis of this result:

\begin{quotation}
In 1972 while I was a PhD student exploring the application of the Stackelberg
Strategy to dynamic games, I had extensive discussions with my advisor
Professor Cruz about my observation and conclusion, at the time, that the
closed-loop controls in Stackelberg dynamic games did not satisfy Bellman's
Principle of Optimality and that the controls had to be recalculated again at
every instant of time along the optimal trajectory. Initially, he had some
doubt about the validity of this result, but after long discussions, we were
convinced. The results were published in a follow up paper on Stackelberg
dynamic games in 1973. Unfortunately, at that time, we did not make the
connection with macroeconomics, but fortunately, several years later, Kydland
and Prescott made the connection and published their famous 1977 paper.

Coincidentally, in 1977 I was a faculty member in electrical engineering at
the University of Pittsburgh, which is right next door to Carnegie Mellon
University. I had no idea that Kydland and Prescott were at CMU and I am not
sure if they were aware that I was at the University of Pittsburgh about five
minutes' walk from their offices. I discovered their work on time
inconsistency several years later and was very happy that they received the
Nobel Prize for it in 2004. (email received 14th July 2020.)
\end{quotation}

The transplant of Simaan and Cruz' (1973b) time-inconsistency of Stackelberg
dynamic games to economics was quickly done by Kydland (1975, 1977),
Prescott's Ph.D.\ student at Carnegie Mellon University, where he completed
his PhD in 1973 on "Decentralized Economic Planning". Cass first encouraged
Kydland to expand an idea about non-cooperative and dynamic games into a
thesis (Kydland (2014)). In a paper based on a chapter of his
PhD\ dissertation, Kydland (1975) acknowledges his marginal contribution with
respect to Simaan and Cruz (1973a, b):

\begin{quotation}
From the viewpoint of game theory this section does not really offer any new
results... The dominant player problem, on the other hand, has only recently
received a little attention in the game literature, and the two interesting
papers by Simaan and Cruz [1973a, 1973b] should be mentioned (Kydland (1975), p.~323).
\end{quotation}

After Kydland's (1975, 1977) citations, Kydland and Prescott (1980 p.80 and
86) and Calvo (1978, equation A.10, p.1425 and A.21 p.1427) use the initial
transversality condition at the origin of time inconsistency without citing
Simaan and Cruz (1983b, equation 11, p.621). Calvo (1978) correctly states that:

\begin{quotation}
We encounter time inconsistency even when the government attempts to maximize
the welfare of the representative individual, that is to say, in a context
where there is not a shade of malevolence or dishonesty at play. (Calvo
(1978), p.~1422)
\end{quotation}

In Simaan and Cruz (1973b), Calvo (1978) and Kydland and Prescott (1980),
there is time inconsistency and the policy maker knows the transmission
mechanism parameters which do not change over time. Therefore, he has unbiased
estimates of these parameters, as opposed to the policy maker in Kydland and
Prescott (1977), section 5, which is not related to dynamic time-inconsistency.

Although the optimal discretionary policy at date $0$ is sub-optimal with
respect to all other optimal discretionary policies with an optimization
starting at any future date, it is always less sub-optimal than rules which
peg policy instrument to their long run value, for dynamic models if $A>0$.

Kydland and Prescott (1980) rename the Stackelberg dynamic game as Ramsey
optimal policy, in reference with Ramsey's (1927) static model of optimal taxation.

\begin{quotation}
In his paper on optimal taxation, Ramsey (1927, p. 59) briefly considers the
dynamic problem, but because it is `considerably more difficult' essentially
assumes the dynamics away. (Kydland and Prescott (1980), p.84)
\end{quotation}

Kendrick (2005) mentions Kydland's talk in 1975:

\begin{quotation}
Since the idea of forward variables like those above was of obvious importance
in dynamics and control, we invited Lucas to give a talk at the Society of
Economic Dynamics and Control (SEDC) conference in Cambridge, MA in 1975. He
declined to come, but Finn Kydland did accept our invitation. Finn's talk at
that meeting was well attended and listened to carefully. The reaction of the
control engineer standing next to me was typical -- \textquotedblleft no
problem -- you just have to treat it like a game theory
problem\textquotedblright. (Kendrick (2005), p.~15)
\end{quotation}

Ljungqvist and Sargent (2012) judge that Prescott's (1977) \textquotedblleft
pessimism\textquotedblright\ on optimal control has been overturned:

\begin{quotation}
Prescott (1977) asserted that recursive optimal control theory does not apply
to problems with this structure. This chapter and chapters 20 and 23 show how
Prescott's pessimism about the inapplicability of optimal control theory has
been overturned by more recent work. The important contribution by Kydland and
Prescott (1980) helped to dissipate Prescott's initial pessimism. (Ljungqvist
and Sargent (2012), chapter 19)
\end{quotation}

Ljungqvist and Sargent (2012) expose the optimal program in the linear
quadratic case of a Stackelberg dynamic game. This is the same optimal program
as the one presented in Simaan and Cruz (1973b), Kydland (1975, 1977), Calvo
(1978) and Kydland and Prescott (1980). \textquotedblleft More recent
work\textquotedblright\ only assumes a policy maker's commitment constraint as
initially proposed by Calvo (1978), published the year after Prescott (1977):

\begin{quotation}
It is clear that no inconsistency arises if the government optimizes at
$t_{0}$, say, and abides by the dictates of that policy for all $t\geq t_{0}$;
so one possible proposal could be constraining the government to do just that
for a given $t_{0}$. (Calvo (1978), p.~1422)
\end{quotation}

The contribution of \textquotedblleft more recent work\textquotedblright\ was
to \emph{forget} Prescott's (1977) and Kydland and Prescott's (1977)
\emph{over-statement} of the inapplicability of optimal control for
stabilization policy.

Finally, how large are the relative gains of re-optimization on future dates
$t>t_{0}$ with respect to initial optimal policy on date $t_{0}$? Firstly, if
the initial shock is small and the policy target remains close to its long run
equilibrium, the optimal path hardly deviates after re-optimization. Secondly,
if the initial shock is large and if one re-optimizes soon after, the optimal
path hardly deviates. Thirdly, if the relative cost of changing the policy
instrument ($R/Q$) is very large or very small, the optimal path hardly
deviates after re-optimization (Chatelain and Ralf (2016)).

If there is a fixed cost of future re-optimization, re-optimization is more
likely to happen (1) following large shocks $\pi_{0}$ although one may
maintain the linear-quadratic approximation (2) in the middle of the duration
of the way back to equilibrium and (3) if $R/Q\approx1$: the cost of changing
the policy instrument is of the same order of magnitude than the cost of the
volatility of the policy target.

\subsection{Barro and Gordon (1983) Steady State Bias in the Policy Maker's
Loss Function}

As opposed to Calvo's (1978) time-inconsistency in dynamic Stackelberg games,
Kydland and Prescott (1977, section 3) and Barro and Gordon (1983a, 1983b)
assume that the policy maker's utility is different from the private sector's
utility (there is a "\emph{shade of malevolence or dishonesty at play}" by the
policy maker) and that the policy transmission mechanism is a static model
($A=0$).

\begin{definition}
Optimal discretion is such that, firstly, one assumes $A=0$ in the policy
transmission mechanism, secondly the policy maker has a long run target value
of the policy target ($\pi_{D}>0$) which is \textbf{systematically different}
from the long-run set point value of the policy target in the transmission
mechanism, contrary to \textquotedblleft rules\textquotedblright.
\end{definition}

The policy target is unemployment $U_{t}$ (Barro and Gordon (1983a)) or output
$x_{t}$ (Barro and Gordon (1983b)) and their policy instrument is inflation.
The transmission parameter is a static Phillips curve: $U_{t}=B\pi_{t}$ with
$B_{U}<0$ or $x_{t}=B_{x}\pi_{t}$ with $B_{x}>0$. \ We translate their model
in our benchmark model with our notations: $\pi_{t}=Bi_{t-1}$ with $B<0$ and
$A=0$.

\begin{proposition}
The joint assumptions $A=0$ (static model), a non-zero weight on the
volatility of th policy target $Q>0$ and a distorted long run objective for
the policy target $\pi_{D}>0$ (inflationary bias) implies that optimal
discretion is necessarily sub-optimal with respect to rules which are optimal
assuming $A=0$ and $\pi_{D}=0$
\end{proposition}

\begin{proof}
Optimal discretionary policy is obtained by minimizing in each period the
expected loss fonction. The first order condition is:%
\begin{equation}
\frac{\partial}{\partial\pi_{t}}\left(  \beta Q\left(  \pi_{t}-\pi_{D}\right)
^{2}+R\left(  \frac{\pi_{t}}{B}\right)  ^{2}\right)  =2\beta Q\left(  \pi
_{t}-\pi_{D}\right)  +2\frac{R}{B^{2}}\pi_{t}=0.\nonumber
\end{equation}
As expected, this leads to a suboptimal solution denoted ($\pi_{D}^{\ast}$,
$i_{D}^{\ast}$) with a policy maker's distorted steady state $\pi_{D}>0$ when
compared to the solution of \textquotedblleft rules\textquotedblright%
\ ($\pi_{D}^{\ast}=0$, $i_{D}^{\ast}=0$) which is optimal for a static model
($A=0$) with a policy maker's non-distorted steady state $\pi_{D}=0$:%
\[
\pi_{D}^{\ast}=\frac{\beta Q}{\beta Q+\frac{R}{B^{2}}}\pi_{D}>0=\pi^{\ast
}\text{ and }i_{D}^{\ast}=\frac{1}{B}\frac{\beta Q}{\beta Q+\frac{R}{B^{2}}%
}\pi_{D}<0=i^{\ast}.
\]

\end{proof}

The discretionary equilibrium of the policy target $\pi_{D}^{\ast}$ increases
with the steady state inflationary bias $\pi_{D}$ in the central bank's
preference and with the weight of the policy target $Q$ and decreases with the
weight on the policy instrument $R$. The same level of $\pi_{D}^{\ast}$ is
obtained for a lower inflationary bias $\pi_{D}$ compensated by a higher $Q/R$
ratio. Assuming a zero weight ($Q=0$, $R>0$, $\pi_{D}\geq0$) on the variance
of the policy target with respect to the policy maker's biased long run
target, is equivalent to assume a non-zero weight and no inflationary bias
($Q\geq0$, $R>0$, $\pi_{D}=0$).

Instead of directly assuming that a conservative central banker has an
exogenously set lower inflationary bias $\pi_{D,C}$ in the loss function
$0<\pi_{D,C}<\pi_{D}$ which would be too simple for a publication in a top
journal, Rogoff's (1985) contribution is to assume that a conservative central
banker has an exogenously lower relative weight on the policy target $Q/R$ for
an unchanged exogenous inflationary bias $\pi_{D}$\ in the loss function.

Barro and Gordon (1983a) mention that if ever the policy maker with a
distorted steady state (exogenous inflationary bias) would choose
\textquotedblleft rules\textquotedblright\ ($\pi^{\ast}=0$), he would have an
incentive to deviate creating surprise inflation. They mention that this
solution does not hold in a rational expectations equilibrium ($\pi_{D}^{\ast
}$, $i_{D}^{\ast}$).

Barro and Gordon (1983a) justify their exogenous steady state inflationary
bias assumption by this single sentence:

\begin{quotation}
In the presence of unemployment compensation, income taxation, and the like,
the natural unemployment rate will tend to exceed the efficient level - that
is, privately chosen quantities of marketable output and employment will tend
to be too low. (Barro and Gordon (1983a), p.~593)
\end{quotation}

Taylor, however, is not convinced by this argument:

\begin{quotation}
In other well-recognized time inconsistency situations, society seems to have
found ways to institute the optional (cooperative) policy... The superiority
of the zero inflation policy is obvious... It is therefore difficult to see
why the zero inflation policy would not be adopted. (Taylor (1983), p.125)
\end{quotation}

Blinder (1998) comments on Barro and Gordon's (1983a) positive interpretation
of their model explaining high inflation during 1973-1979:

\begin{quotation}
Barro and Gordon ignored the obvious practical explanations for the observed
upsurge in inflation -- the Vietnam War, the end of the Bretton-Woods system,
two OPEC shocks, and so on -- and sought instead a theoretical explanation for
what they believed to be a systematic inflationary bias in the behaviour of
central banks. (Blinder (1998), p.40).
\end{quotation}

A\ version of Barro and Gordon (1983a) was soon included in graduate textbooks
of macroeconomics in order to explain Kydland and Prescott's (1977) time
inconsistency. For now four decades, which included a global disinflationary
trend and the emergence of inflation targeting, graduate students around the
world have exams related to the unlikely statement that central bankers have a
pro-inflation bias in their loss function:

\begin{quotation}
Though it didn't hurt them in academia, Kydland and Prescott's timing was
exceptionally poor from a real world perspective. As mentioned earlier, the
late 1970s and early 1980s witnessed sharp and painful disinflations in the
US, the UK, and elsewhere. Neither Paul Volcker nor Margaret Thatcher (the
Bank of England was not independent then) succumbed to the temptation posed by
time inconsistency; they probably never even heard of it. It turned out that
Kydland, Prescott, and other academics were prescribing how to fight the last
war just as the next war was getting underway. (Blinder (2020), p.31).
\end{quotation}

\section{Taylor (1993) Translates the Rules \emph{versus} Discretion
Controversy}

The perception of rules \emph{versus} discretion changed fundamentally with
the introduction of the Taylor rule in the 1990s. Like Barnett and Serletis
(2017), Kendrick (2005), and Turnovsky (2011), Taylor confirms the gains of
the proponents of the policy ineffectiveness claim in the 1980s:

\begin{quotation}
His (Ben McCallum (1999)) lecture also describes how research work on monetary
policy rules waned considerably in the 1980s, except for the work of small
groups \textquotedblleft toiling in the vineyards\textquotedblright. I call it
\textquotedblleft dark ages\textquotedblright\ in another paper; it seemed
like everyone interested in the new rational expectations methods in the 1980s
was working on real business cycle models without a role for monetary policy.
(Leeson and Taylor (2012))
\end{quotation}

Taylor refers to mainstream macroeconomists belonging to top US\ academic
circles in the 1980s, for example, the ones invited to NBER\ conferences, as
well as their PhD\ students. By contrast, in the Fed and in other central
banks, policy advisers still used models involving policy maker's negative
feedback response. But they had less and less access to publishing their
results in top academic journals. In continental Europe, in particular in
France and Belgium, macroeconomists around Dr\`{e}ze, Benassy and Malinvaud
were involved in the heyday of disequilibrium macroeconomics in the 1980s.
They considered that real business cycles models were not relevant.

During the \textquotedblleft dark ages\textquotedblright of the 1980s, Taylor
followed the strategies of dissenters in scientific controversies with the
outcome described by Latour (1987, p.137): \textquotedblleft\emph{From a few
helpless occupying a few weak points, they end up controlling strongholds.}%
\textquotedblright\ His strategy targeted policy makers at the Fed, in order
to convince them that ``rules" can also be a specific ``feedback rules".

One may be surprised that Taylor (1993) dedicated the first section of his
paper to semantic issues. Semantic issues matter in a scientific controversy
in other sciences than mathematics (which always proceed with exact
mathematical definitions).

Friedman (1948) and Kydland and Prescott (1977) originally \emph{translated}
``stabilization policy ineffectiveness" and ``the fixed setting of the policy
instrument" \emph{into} ``rules". They translated ``negative feedback policy"
\emph{into} ``discretion". The word ``rules" is related to more virtues than
``discretion", which has a wide range of negative polysemic connotations such
as the abuse of power or arbitrary, opportunistic, careless, thoughtless or
non-predictable random behavior.

For the proponents of stabilization policy, these negative connotations are
unrelated to the expected virtues of stabilization policy aiming to improve
social welfare. For example, Buiter (1981) translated \textquotedblleft
negative feedback policy" into the oxymoron \textquotedblleft flexible rules":

\begin{quotation}
This paper analyses an old controversy in macroeconomic theory and policy:
'rules versus discretion' or, more accurately, fixed rules (rules without
feedback or open-loop rules) versus flexible rules (contingent rules,
conditional rules, rules with feedback or closed-loop rules). (Buiter (1981), p.~647)
\end{quotation}

In their textbook, Dornbush and Fischer translated ``negative feedback policy"
into ``activist monetary rule":

\begin{quotation}
The growth rate of money is high when unemployment is high and is low when
unemployment is low. That way, monetary policy is expansionary at times of
recession and contractionary in a boom. (Dornbush and Fischer (1984), p.~343)
\end{quotation}

The fixed setting of the policy instrument is translated into
\textquotedblleft\emph{inactive policy}" (Blanchard and Fischer (1989), p.582)
or \emph{passive policy}, which may carry the negative connotations of policy
makers who may be lazy and useless bureaucrats being paid for doing nothing.

The \emph{translation} of the definition of the word \textquotedblleft
rule\textquotedblright\ by \emph{its opposite} in the controversy is a
dissenters' tactical move:

\begin{quotation}
Tactic four: rendering the detour invisible... People can still \emph{see} the
difference between what they wanted and what they got, they can still feel
they have been cheated. A fourth move is then necessary that turns the detour
into a progressive drift [e.g. \textquotedblleft activist
policy\textquotedblright, \textquotedblleft systematic
policies\textquotedblright], so that the enrolled group still thinks that it
is going along a \emph{straight} line, without ever abandoning its own
interests... I should now be clear why I used the word \emph{translation}. In
addition to its linguistic meaning (relating versions in one language to
versions in another one), it has also a geometric meaning (moving from one
place to another). Translating interests is at once offering \emph{new
interpretations }of these interests and channeling people in other directions.
(Latour (1987), p.~116-117)
\end{quotation}

In the Harry Johnson lecture at the Money Macro and Finance Research Group
conference in Durham, September 1997, Taylor (1998) defines \textquotedblleft%
\emph{translational economics}\textquotedblright\ in a more narrow sense than
Latour (1987):

\begin{quotation}
The process of finding ways for the research in the biology laboratories to be
applied in medicine to improve people's health is called `translational
biology'. Analogously the term translational economics might usefully
designate the process of finding ways to make academic research in economics
applicable to improving the performance of an economy. (Taylor (1998), p.~7)
\end{quotation}

Taylor used his position in the Council of Economic Advisers from June 1989 to
August 1991 during the George H.W Bush administration:

\begin{quotation}
As a member of the CEA I felt that I had an opportunity to move the policy
`ball' at least a little bit in the direction of the policy rule `goal line',
an opportunity that would not exist outside the policy arena. One plan of
action was to use the public forum offered by the annual Economic Report of
the President [1990] to make the case for monetary policy rules. (Taylor
(1998), p.~8-9)
\end{quotation}

In a first step, Taylor translated Keynesian \textquotedblleft negative
feedback policy\textquotedblright\ into \textquotedblleft systematic
policy\textquotedblright:

\begin{quotation}
When we wrote about policy \textquotedblleft rules\textquotedblright\ in the
1990 \emph{Economic Report to the President}, for example, we said
\textquotedblleft systematic\textquotedblright\ policies instead of rules so
as not to confuse people who might think that a rule meant a fixed setting for
the policy instrument. (Leeson and Taylor (2012))
\end{quotation}

McCallum (2015) mentions:

\begin{quotation}
At the time I was a member of the Carnegie-Rochester Conference's advisory
board. One of our duties was to suggest fruitful topics for future
conferences. In that capacity, at the board's planning meeting for the
November 1992 conference, some months earlier, I had suggested that a paper
should be commissioned that would develop some method or criterion that would
permit an outside researcher to determine whether an actual central bank's
actions over some significant span of time should be regarded as resulting
from a policy rule, rather than being \textquotedblleft
discretionary\textquotedblright. This suggestion met with approval by Allan
Meltzer and the other participants, and then we quickly agreed that the best
person to write the paper would be John Taylor---who Allan subsequently
contacted and signed up for the conference.

Well, as it turns out, the paper that John delivered did not actually do this.
Instead it proposed and (very effectively) promoted a specific rule. So at the
time of the conference, I, evidently, must have been somewhat put off by the
change in focus. (This change might have been arranged with Allan; about that
I do not know.). (McCallum (2015), p.2).
\end{quotation}

Secondly, in Taylor (1993) first section \textquotedblleft Semantic
issues\textquotedblright\ of a paper first presented in November 1992 at a
meeting in Pittsburgh of the Carnegie-Rochester conference on public policy,
and then published in the resulting conference volume in the spring of 1993,
Taylor (1993) translated the \textquotedblleft negative feedback
policy\textquotedblright\ into its opposite in the Friedman (1948)
controversy: \textquotedblleft policy rules\textquotedblright.

\begin{quotation}
There is considerable agreement among economists that a policy rule need not
be interpreted narrowly as entailing fixed settings for the policy
instruments. Although the classic rules versus discretion debate was usually
carried on as if the only policy rule were the constant growth rate rule for
the money supply, feedback rules in which the money supply responds to changes
in unemployment or inflation are also policy rules... A policy rule is a
contingency plan that lasts forever unless there is an explicit cancellation
clause. (Taylor (1993), p.~198).
\end{quotation}

Taylor (1993) proposed the following new definitions of rules \emph{versus}
discretion. He fixed rounded values of the parameters of a feedback rule which
was roughly predicting the Federal funds rate in the last six years (1987-1992):

\begin{definition}
Taylor (1998, 2017): \textquotedblleft Rules\textquotedblright\ follow
systematically a Taylor (1993) negative-feedback \textquotedblleft
rule\textquotedblright\ where the real federal funds rates has a 2\% set point
and responds in proportion to inflation deviation from 2\% and to output gap
with given parameters equal to 1/2.%
\[
i_{t}=\pi_{t}+\frac{1}{2}\left(  \pi_{t}-2\right)  +\frac{1}{2}x_{t}%
+2=1.5\pi_{t}+0.5x_{t}+1.
\]

\end{definition}

\begin{definition}
Taylor (1998, 2017): \textquotedblleft Discretion\textquotedblright\ is any
other monetary policy which deviates from a Taylor (1993) negative-feedback
rule, measured by the discrepancies (or the residuals) of funds rate with
respect to Taylor (1993) rule predictions.
\end{definition}

\begin{quotation}
The discrepancies between the [Taylor rule] equation and reality could be a
measure of discretion, either for good or for bad. (Taylor (1998), p.~12)
\end{quotation}

Taylor's (1993) policy paper was using a rhetorical presentation in order to
convince the Fed's policy makers, not academics. A feedback rule with fixed
parameters can be easily understood. His rule had soon early users among a
sufficiently large set of influential practitioners (in this case, at the Fed)
and of private sector analysts of policy making (Koenig et al. (2012)).

Taylor's (1993) plots $24$\ quarterly observations of Fed funds rate from
$1987$ to $1992$ and its forecast according to his rule. No statistical tests
are presented. The rounded parameters of his rule seem arbitrary (the formal
constraint is that they involve only the numbers $1$\ and $2$). There is no
information on the transmission mechanism, which is a crucial information for
designing negative feedback policy rule. There is no mention of reverse
causality and nor of endogeneity issues related to this forecast. The paper
was published in a conference volume by invitation along with comments by
McCallum (1993), so that peer-review was limited.

We mention Lucas, Kydland and Prescott overstated their results in the rules
versus discretion controversy. Taylor (1993) may have overstated the response
of Fed's fund rate to inflation during 1987-1992 in his policy rule. Let's
check the robustness of Taylor's (1993) rule estimating this policy rule with
currently available and revised quarterly data from FRED database for
inflation and congressional budget office (CBO) output gap, still for
$24$\ observations from $1987$ to $1992$. The inflation rule parameter
estimate is below one ($0.75$) instead of $1.5$ in the Taylor rule, so that
the Taylor principle is not satisfied. The output gap rule parameter estimate
is $0.75$ instead of $0.5$. The share of unexplained variance related to
discretion is $30\%$.

Inflation has a very low persistence with a very low autocorrelation
coefficient: $\rho_{\pi}=0.4$. By contrast, Fed funds rate and CBO output gap
are highly persistent, close to unit root with $\rho_{i}=0.96$ and $\rho
_{x}=0.935$. A regression of Fed funds rate on its lagged value and current
output gap leads to an $R^{2}=962\%$: the share of unexplained variance of the
policy rule related to discretion is less than $4\%$. Including inflation in
the regression only increases the $R^{2}$ by $0.0013\%$, which may correspond
to overfitting. Inflation persistence is low because of private sector's
negative feedback mechanism and because of a divine coincidence where leaning
against output gap fluctuations stabilizes the component of inflation
correlated with output gap. A policy maker's instrument should respond more to
more persistent target variables (Ashley \textit{et al.} (2020)). Rounding the
figures of the estimates in order to introduce the number $2$\ for a key
parameter, one may have proposed an alternative definition of "rules":

\begin{definition}
Rules follows systematically a persistence-dependent negative feedback rule
where Fed funds rate responds to its lagged value and to current output gap,
with Fed's funds rate long run sensitivity to the output gap equal to $2$.
\[
i_{t}=0.8i_{t-1}+0.4x_{t}\text{ with }2=\frac{0.4}{1-0.8}%
\]

\end{definition}

\begin{definition}
\textquotedblleft Discretion\textquotedblright\ is any other monetary policy
which deviates from this persistence negative-feedback rule, measured by the
discrepancies of funds rate with respect to the forecast of
persistence-dependent rule. For the period 1987-1992, discretion corresponds
to the unexplained share of variance is around $5\%$ in the estimation of the
policy rule.
\end{definition}

The unexpected side-effect of the interest paid by policy makers to Taylor's
interest rate rule was fed back three years later into academia with citations
in academic papers starting in 1996. These papers soon led to the emergence of
the three equation new-Keynesian theoretical model where the Taylor rule is
the third equation (Kerr and King (1996), Clarida, Gali and Gertler (1999)).
Google Scholar citations of Taylor (1993) increased by at least $40$ citations
each year from 1997 to 2008. The success of the Taylor (1993) rule tilted the
balance to formerly \textquotedblleft discretion\textquotedblright\ in the
original \textquotedblleft rules \emph{versus} discretion\textquotedblright%
\ controversy, with the gain of the original negative feedback
\textquotedblleft discretion" being renamed \textquotedblleft rule".

\begin{quotation}
Essentially, when the usefulness of discretion was rediscovered, the notion of
rule got the credit... This peculiar \emph{re-definition} of `rule' not only
renders the terminology bogus, but also the substance of rules vs discretion
Mark I. (Bibow (2004), p.~559).
\end{quotation}

The Taylor rule marked the return of classic control into top level
macroeconomic research. However, advisors using models at the Fed never
stopped using feedback rules. Even though optimal rules were considered too
complex to be disclosed to the press, forty years later after the initial use
of optimal control with the Fed FMP model, Fed's optimal control approach in
the FRB/US model made the headlines of financial newspapers, following
Yellen's (2012) speech. Yellen was at the time Vice Chair of the board of
governors of the Federal Reserve System:

\begin{quotation}
First, the FRB/US model's projections of real activity, inflation, and
interest rates are adjusted to replicate the baseline forecast values. Second,
a search procedure is used to solve for the path of the federal funds rate
that minimizes the value of an assumed loss function, allowing for feedback of
changes in the federal funds rate from baseline to real activity and
inflation. For the purposes of the exercise, the loss function is equal to the
cumulative discounted sum from 2012:Q2 through 2025:Q4 of three factors--the
squared deviation of the unemployment rate from 5-1/2 percent, the squared
deviation of overall PCE inflation from 2 percent, and the squared quarterly
change in the federal funds rate. The third term is added to damp
quarter-to-quarter movements in interest rates. (Yellen (2012), p.~16).
\end{quotation}

\section{Conclusion}

The present paper studied the use of control theory in macroeconomic models
with a special emphasis on the debate of monetary policy and -- in particular
-- the question of \textquotedblleft rules \textit{versus}
discretion\textquotedblright. Whereas the 1960s were dominated by a belief in
the effectiveness of discretionary stabilization policy, the rhetorics changed
in the subsequent years. The Lucas (1976) critique and Kydland and Prescott's
(1977) time inconsistency convinced a sufficient number of macroeconomists
that \emph{the negative-feedback mechanism of classic and optimal control
cannot be used for stabilization policy}, whereas it should be used to model
the private sector. Using rhetorical over-statements in their final
conclusions, helped to convince their readers of the \textquotedblleft
stabilization policy ineffectiveness claim". These influential papers paved
the way to Kydland and Prescott's (1982) real business cycle theory, assuming
that stabilization policy had no effect on US postwar business cycles.

Furthermore, these influential papers succeeded in delaying the expected fast
transplant of the robust control, linear quadratic Gaussian estimations and
Ramsey optimal policy in the 1980s to nearly twenty years. Our analysis stops
in 1993 with the emergence of the Taylor (1993) rule, which contributed to the
revision and the decline of the stabilization policy ineffectiveness claim. At
the end of the nineties, negative-feedback rules in backward-looking models
were used at the Fed, and Ramsey optimal policy and robust optimal control
gained interest in monetary policy analysis.

Also in the mid-1980s, however, another controversy started, namely the debate
\textquotedblleft commitment \textit{versus} discretion\textquotedblright. As
a starting point the model of Oudiz and Sachs (1985) can be seen. They modeled
a policy maker's behavior who re-optimizes in each period, as if he is
replaced by another policy maker at the end of the period. As this policy
maker lives only one period, he sets a probability equal to zero to the
expectations of the policy target. Omitting expectations, he does a static
optimization of an otherwise dynamic process for the policy target. Phillips
(1954) already mentioned that static optimization, when wrongly applied to a
dynamic model, can lead to a positive feedback rule which implies exploding
paths of the policy target. Therefore, a set of additional mathematical
restrictions are added on the private sector's behavior in order to eliminate
these unstable paths (Chatelain and Ralf (2021)). Oudiz and Sachs (1985)
solution was initially labeled \textquotedblleft time-consistent" policy. At
the beginning of the 1990s, it became the new benchmark model describing
\textquotedblleft discretion" (Clarida, Gali, Gertler (1999)).

The shift from the original definitions of ``rules" versus ``discretion" to
these models of ``commitment" versus ``discretion" brought semantic issues.
Firstly, the mathematics of this new ``discretion" model implies a
state-contingent positive feedback ``rule". Conversely, Ramsey optimal policy
under commitment has a state-contingent negative feedback ``rule", which
corresponds to the original definition of ``discretion" by Friedman (1960).

When discussing discretionary policies, a macroeconomist involved into policy
making at the end of the seventies may misunderstand a macroeconomist having
done his PhD\ at the end of the nineties. The older macroeconomist will
describe Volcker's policy as discretionary policy, with a strong negative
feedback effect against inflation. The next generation macroeconomist will
describe Volcker's policy as commitment (able to sharply decrease the
expectations of inflation) instead of discretion.

In a second approach, Currie and Levine (1985) defined a \textquotedblleft
simple rule" dynamic path including a proportional feedback rule when the
policy target is a jump variable without an initial condition. Although the
policy maker does not optimize, this equilibrium solution can be a reduced
form of Oudiz and Sachs' (1985) \textquotedblleft discretion" model (Chatelain
and Ralf (2020b)). Their \textquotedblleft simple rule" solution also implies
positive feedback rule parameters and exploding paths of the policy target. As
in Oudiz and Sachs' (1985) solution, a set of additional mathematical
restrictions on the private sector's behavior are assumed in order to
eliminate these unstable paths (Chatelain and Ralf (2020a and 2020b)).

\textquotedblleft Time-consistent" and \textquotedblleft simple rule" dynamic
paths with positive-feedback rule parameters advocate the \emph{opposite of
control} which favors negative feedback rule parameters for stabilization
policy as obtained with Ramsey optimal policy under commitment. Other than in
the above discussed \textquotedblleft rules\textquotedblright\ paradigm with
no policy feedback in Friedman (1960), Kydland and Prescott (1977, 1982),
Barro and Gordon (1983a)) due to confidence in the stabilizing force of the
private sector, in the subsequent \textquotedblleft time-consistent" or
\textquotedblleft simple rule" paradigm, positive feedback rules were favored
with a set of mathematical restrictions on the private sector's behavior,
eliminating unstable paths, as in Clarida, Gali, Gertler (1999), among many
other papers. \emph{In both cases, macroeconomists lost control. }The latter
story, however, deserves to be treated in a separate paper.


\begin{thebibliography}{999}                                                                                              %


\bibitem {Aoki 2}Aoki, M. (Ed.). (1967). \textit{Optimization of stochastic
systems: topics in discrete-time systems} (Vol. 32). Academic Press, New York
-- London.

\bibitem {Aoki}Aoki, M. (1975). On a generalization of Tinbergen's condition
in the theory of policy to dynamic models. \textit{Review of Economic
Studies}, 42, 293--296.

\bibitem {Ashley}Ashley, R., Tsang, K. P., \& Verbrugge, R. (2020). A new look
at historical monetary policy (and the great inflation) through the lens of a
persistence-dependent policy rule. Oxford Economic Papers, 72(3), 672-691.

\bibitem {Astrom Kumar}Astr\"{o}m K.~J. and Kumar P.~R. (2014). Control: A
Perspective. \textit{Automatica }50, 3--43.

\bibitem {Athans}Athans, M., Chow, G.~C., 1972. Introduction to stochastic
control and economic systems. \textit{Annals of Economic and Social
Measurement.} 1, 375--384.

\bibitem {Barnett Serletis}Barnett W.~A. and Serletis, A. (2017). An Interview
with William A. Barnett. \textit{Econometrics}, ISSN 2225-1146, MDPI, Basel.
5(4), 1--32.

\bibitem {Barnett}Barnett, W.~A. (2017). Collaboration with and without
Coauthorship: Rocket Science Versus Economic Science. In \textit{Collaborative
Research in Economics}, 137--151. Palgrave Macmillan, Cham.

\bibitem {Bennett}Bennett, S. (1996). A brief history of automatic control.
\textit{IEEE Control Systems Magazine} 16(3), 17--25.

\bibitem {Barro}Barro, R.~J. (2009). Rare disasters, asset prices, and welfare
costs. \textit{American Economic Review} 99(1), 243--64.

\bibitem {Barro Gordon}Barro R.~J. et Gordon D.~B. (1983a), A positive theory
of monetary policy in a natural rate model, \textit{Journal of Political
Economy} 91(4), 589--610.

\bibitem {Barro Gordon 2}Barro, R.~J., \& Gordon, D.~B. (1983b). Rules,
discretion and reputation in a model of monetary policy. \textit{Journal of
Monetary Economics} 12(1), 101--121.

\bibitem {Bissel}Bissell, C. (2010). Not just Norbert. \textit{Kybernetes}
39(4), 496--509.

\bibitem {Bibow}Bibow, J. (2004). Reflections on the current fashion for
central bank independence. \textit{Cambridge Journal of Economics} 28(4), 549--576.

\bibitem {Blanchard Fischer}Blanchard, O.~ J. and Fischer, S. (1989).
\textit{Lectures on macroeconomics}. MIT Press, Cambridge (Mass.) -- London.

\bibitem {Blanchard Kahn}Blanchard O.~J. and Kahn C.\ (1980). The solution of
linear difference models under rational expectations. \textit{Econometrica}
48, 1305--1311.

\bibitem {Blinder 1987}Blinder, A.~S. (1987). Keynes, Lucas, and scientific
progress. \textit{The American Economic Review} 77(2), 130--136.

\bibitem {Blinder}Blinder A.~S. (1998). \textit{Central Banking in Theory and
practice}, MIT\ Press, Cambridge (Mass.).

\bibitem {Blinder 2020}Blinder A.~S. (2020). What does Jerome Powell know that
William McChesney Martin didn't---and what role did academic research play in
that? \textit{Manchester School}.\ 88(S1), 32-49.

\bibitem {Boumans}Boumans, M. (2020). The Engineering Tools that Shaped the
Rational Expectations Revolution. \textit{History of Political Economy} 52(1), 143-167.

\bibitem {Brainard}Brainard, W.~C. (1967). Uncertainty and the Effectiveness
of Policy. \textit{The American Economic Review} 57(2), 411--425.

\bibitem {Buiter}Buiter, W.~H. (1981). The superiority of contingent rules
over fixed rules in models with rational expectations. \textit{The Economic
Journal} 91(363), 647--670.

\bibitem {Calvo}Calvo, G.~A. (1978). On the Time Consistency of Optimal Policy
in a Monetary Economy. \textit{Econometrica} 46(6), 1411--1428.

\bibitem {Chatelain Ralf}Chatelain, J.~B. and Ralf, K. (2016). Countercyclical
versus Procyclical Taylor Principles. Available at SSRN 2754054.

\bibitem {Chatelain Ralf 2}Chatelain, J.~B. and Ralf, K. (2020a). Ramsey
Optimal Policy versus Multiple Equilibria with Fiscal and Monetary
Interactions. \textit{Economics Bulletin} 40(1), pp. 140-147.

\bibitem {Chatelain Ralf 4}Chatelain, J.~B. and Ralf, K. (2020b). Policy
Maker's Credibility with Predetermined Instruments for Forward-Looking
Targets. \textit{Revue d'Economie Politique} 129, forthcoming.

\bibitem {Chatelain Ralf 3}Chatelain, J.~B. and Ralf, K. (2021). Imperfect
Credibility versus No Credibility of Optimal Monetary Policy, \textit{Revue
Economique} 72(1), forthcoming.

\bibitem {Chen }Chen C.~I., Cruz J.~B. (1972). Stackelberg Solution for
Two-Person Games with Biased Information Patterns. \textit{IEEE Trans. on
Automatic Control} AC-17 (6), 791--798.

\bibitem {Clarida Gali Gertler}Clarida R., Gali J., Gertler M. (2000]).
Monetary Policy Rules and Macroeconomic Stability. \textit{Quarterly Journal
of Economics} 115, 147--180.

\bibitem {Currie Levine}Currie, D. and Levine, P. (1985). Macroeconomic policy
design in an interdependent world. In \textit{International economic policy
coordination}. Cambridge University Press, Cambridge, 228--273.

\bibitem {Cruz}Cruz J.~B. (2019). Utilizing data-based artificial intelligence
to enable science-based models and dynamic feedback controllers to adapt to
disasters. \textit{Philippine Science Letters} 12(1), 82--91.

\bibitem {Cruz 75}Cruz J.~B. (1975). Survey of Nash and Stackelberg
Equilibrium Strategies in Dynamic Games. \textit{Annals of Economic and Social
Measurement} 4, 339--344.

\bibitem {Cruz 78}Cruz J.~B. (1978). Leader-Follower Strategies for Multilevel
Systems. \textit{IEEE Trans. on Automatic Control} AC-23, 244--255.

\bibitem {De Rosnay}De Rosnay (1979). \textit{The Macroscope: a new world
scientific system}. English translation. Harper \& Row, Publishers, New York.

\bibitem {Descartes}Descartes, R.(1641). \textit{Meditationes de Prima
Philosophia, in qua Dei existentia et anim\ae \ immortalitas demonstratur.}
Michaelem Soly, Paris.

\bibitem {Dornbusch}Dornbusch, R. and Fischer, S. (1984).
\textit{Macroeconomics.} McGraw-Hill. Inc., New York.

\bibitem {Doyle}Doyle, J.~C. (1978). Guaranteed margins for LQG regulators.
\textit{IEEE Transactions on Automatic Control} 23(4), 756--757.

\bibitem {Duarte}Duarte P.~G. (2009). A Feasible and Objective Concept of
Optimal Monetary Policy: the Quadratic Loss Function in the Postwar Period.
\textit{History of Political Economy} 41(1), 1--55.

\bibitem {Ezekiel}Ezekiel, M. (1938). The cobweb theorem. \textit{The
Quarterly Journal of Economics} 52(2), 255--280.

\bibitem {Forder}Forder, J. and Monnery, H. (2019). Why did Milton Friedman
win the Nobel Prize? a consideration of his early work on stabilization
policy. \textit{Econ Journal Watch} 16(1), 130--145.

\bibitem {Friedman 1948}Friedman, M. (1948). A Monetary and Fiscal Framework
for Economie Stability. \textit{American Economic Review} 38(3), 245--263.

\bibitem {Friedman 1953}Friedman, M. (1953). The effects of a full-employment
policy on economic stability: A formal analysis. In \textit{Essays in positive
economics}. University of Chicago Press, Chicago, 117--132.

\bibitem {Friedman}Friedman M. (1960). \textit{A Program for Monetary
Stability}. New York: Fordham University Press.

\bibitem {Fuhrer}Fuhrer, J.~C. (2010). Inflation persistence. In
\textit{Handbook of monetary economics} Vol. 3. Elsevier, 423--486.

\bibitem {Goutsmedt}Goutsmedt, A., Pinz\'{o}n-Fuchs, E., Renault, M., and
Sergi, F. (2019). Reacting to the Lucas Critique: The Keynesians' Replies.
\textit{History of Political Economy} 51(3), 535--556.

\bibitem {Hansen Sargent}Hansen L.~P. and Sargent T. (2007).
\textit{Robustness}, Princeton University Press, Princeton.

\bibitem {Hansen Sargent2}Hansen L.~P. and Sargent T. (2011). Wanting
Robustness in Macroeconomics. In \textit{Handbook of Monetary Economics, }vol
3(B), Friedman B.M. and Woodford M. editors, Elsevier B.V., 1097--1155.

\bibitem {Hartley}Hartley, J.~E. (2006). Kydland and Prescott's Nobel Prize:
the methodology of time consistency and real business cycle models.
\textit{Review of Political Economy} 18(1), 1--28.

\bibitem {Hayes}Hayes, B. (2011). Economics, Control Theory and the Phillips
Machine. \textit{Economia Politica} 28(1), 83--96.

\bibitem {Ho}Ho Y.~C. (1970). Differential Games, Dynamic Optimization, and
Generalized Control Theory. \textit{Journal of Optimization Theory and
Applications} 6(3).

\bibitem {Holt 62}Holt, C.~C. (1962). Linear decision rules for economic
stabilization and growth. \textit{Quarterly Journal of Economics} 76, 20--45.

\bibitem {Holt 2}Holt, C.~C., Modigliani, F., \& Simon, H. A. (1955). A linear
decision rule for production and employment scheduling. \textit{Management
Science}, 2(1), 1--30.

\bibitem {Holt3}Holt, C.~C., Modigliani, F., \& Muth, J.~F. (1956). Derivation
of a linear decision rule for production and employment. \textit{Management
Science} 2(2), 159--177.

\bibitem {IngramLeeper}Ingram, B. and Leeper E.~M. (1990). \textit{Post}
Econometric Policy Evaluation: A Critique. Board of Governors of the Federal
Reserve System International Finance Discussion Papers No. 393.

\bibitem {Isaacs}Isaacs R.~P. (1965). \textit{Differential Games: a
Mathematical Theory with Applications to Warfare and Pursuit, Control and
Optimization.} New York: John Wiley and Sons.

\bibitem {Kalman 1960}Kalman R.~E. (1960a). Contributions to the Theory of
Optimal Control. \textit{Boletin de la Sociedad Matematica Mexicana} 5, 102-109.

\bibitem {Kalman 60b}Kalman, R.~E. (1960b). A New Approach to Linear Filtering
and Prediction Problems, Transactions of the ASME. \textit{Journal of Basic
Engineering} 82, 35--45.

\bibitem {Kendrick}Kendrick, D.~A. (1976). Applications of control theory to
macroeconomics. In \textit{Annals of Economic and Social Measurement} Volume
5, number 2 (171--190). NBER.

\bibitem {Kendrick 2}Kendrick, D.~A. (2005). Stochastic control for economic
models: past, present and the paths ahead. \textit{Journal of Economic
Dynamics and Control} 29(1-2), 3--30.

\bibitem {Keynes}Keynes, J.~M. (1936). \textit{The General Theory of
Employment, Interest, and Money.} Macmillan, London.

\bibitem {Klamer}Klamer, A. (1984). \textit{The New Classical Macroeconomics.
Conversations with the New Classical Economists and Opponents}. Brighton:
Wheatsheaf Books.

\bibitem {Klein J}Klein J.~L. (2015). The Cold War Hot House for Modeling
Strategies at the Carnegie Institute of Technology. Inet working paper. 19. SSRN.

\bibitem {Koenig}Koenig E.~F., Leeson R., and Kahn G.~A. (Ed.), (2012).
\textit{The Taylor Rule and the Transformation of Monetary Policy}. Hoover
Institution Press, Stanford.

\bibitem {Koopmans}Koopmans, T.~C. (Ed.). (1950). \textit{Statistical
inference in dynamic economic models} (No. 10). Wiley, Hoboken.

\bibitem {Kerr}King, R.~G. and Kerr, W. (1996). Limits on interest rate rules
in the IS model. \textit{FRB Richmond Economic Quarterly} 82(2), 47-75.

\bibitem {Kydland 1}Kydland, F.~E. (1975). Noncooperative and dominant player
solutions in discrete dynamic games. \textit{International Economic Review}
16(2), 321--335.

\bibitem {Kydland2}Kydland, F.~E. (1977). Equilibrium solutions in dynamic
dominant-player models. \textit{Journal of Economic Theory} 15(2), 307--324.

\bibitem {KP1}Kydland, F.~E. and Prescott, E.~C. (1977). Rules rather than
discretion: The inconsistency of optimal plans. \textit{Journal of Political
Economy} 85(3), 473--491.

\bibitem {KP3}Kydland, F.~E. and Prescott, E. C. (1980). Dynamic Optimal
Taxation, Rational Expectations and Optimal Control. \textit{Journal of
Economic Dynamics and Control} 2(1), 79--91.

\bibitem {KP2}Kydland, F.~E. and Prescott, E.~C. (1982). Time to build and
aggregate fluctuations. \textit{Econometrica} 50, 1345--1370.

\bibitem {KydlandPrescott1988}Kydland, F.~E. and E.~C. Prescott (1988). The
workweek of capital and its cyclical implications, \textit{Journal of Monetary
Economics} 21, 343--360.

\bibitem {Kydland}Kydland, F.~E. (2014). On policy consistency. In
\textit{Economics for the Curious} (pp. 77-95). Palgrave Macmillan, London.

\bibitem {Latour}Latour, B. (1987). \textit{Science in action: How to follow
scientists and engineers through society.} Harvard university press, Boston.

\bibitem {Leeson}Leeson, R. (2011). \textit{A.~W.~H. Phillips: Collected Works
in Contemporary Perspective}. Cambridge University Press, Cambridge.

\bibitem {Leeson Taylor}Leeson, R. and Taylor, J.~B.(2012). The Pursuit of
Policy Rules - A Conversation between Robert Leeson and John B.~Taylor, Book
Chapters, in: E.~F. Koenig, R.~Leeson, and G.~A. Kahn (ed.), \textit{The
Taylor Rule and the Transformation of Monetary Policy} chapter 17 Hoover
Institution, Stanford, 309--324.

\bibitem {Ljungqvist}Ljungqvist L. and Sargent T.~J. (2012). \textit{Recursive
Macroeconomic Theory}. 3rd edition. The MIT\ Press. Cambridge, Massaschussets.

\bibitem {Lucas 1975}Lucas, R.~E. (1975). An equilibrium model of the business
cycle, \textit{Journal of Political Economy} 83,~1113--1144.

\bibitem {Lucas}Lucas, R.~E. (1976). Econometric policy evaluation: A
critique. \textit{Carnegie-Rochester conference series on public policy.} 1, 19-46.

\bibitem {Lucas1987}Lucas, R.~E. (1987). \textit{Models of business cycles}.
Basil Blackwell, Oxford.

\bibitem {Lucas 2003}Lucas, R.~E. (2003). Macroeconomic priorities.
\textit{American Economic Review} 93(1), 1--14.

\bibitem {Mayr}Mayr, O. (1971). Adam Smith and the concept of the feedback
system: Economic thought and technology in 18th-century Britain.
\textit{Technology and culture} 12(1), 1-22.

\bibitem {McCallum}McCallum, B. T. (1993). Discretion versus policy rules in
practice: two critical points: A comment. In Carnegie-Rochester Conference
Series on Public Policy, Vol. 39, pp. 215-220. North-Holland.

\bibitem {McCallum 2}McCallum, B.T. (1999), Recent Developments in the
Analysis of Monetary Policy Rules. \textit{Review, Federal Reserve Bank of St.
Louis}, November/December 1999.

\bibitem {McCallum 3}McCallum, B.T. (2015), Remarks on John Taylor's
contributions. In Carnegie Mellon University. Shadow Open Market Committee Meeting.

\bibitem {Meadows}Meadows, D.~H. (2008). \textit{Thinking in systems: A
primer.} chelsea green publishing.

\bibitem {Muth}Muth, J.~F. (1960). Optimal properties of exponentially
weighted forecasts. \textit{Journal of the American Statistical Association}
55(290), 299--306.

\bibitem {Neck}Neck, R. (2009). Control theory and economic policy: Balance
and perspectives. \textit{Annual Reviews in Control} 33(1), 79--88.

\bibitem {Nelson}Nelson, E. (2008). Friedman and taylor on monetary policy
rules: a comparison. \textit{Review of the Federal reserve bank of Saint
Louis} 90(2), 95.

\bibitem {Oudiz Sachs}Oudiz, G. and Sachs, J. (1985). International policy
coordination in dynamic macroeconomic models. In \textit{International
economic policy coordination} Cambridge University Press, Cambridge, 274--330.

\bibitem {Phillips 3}Phillips, A.~W. (1954a), The Mechanism of Economic
Systems [book review], \textit{The Economic Journal} 64, 805--807.

\bibitem {Phillips}Phillips, A.~W. (1954b). Stabilisation policy in a closed
economy. \textit{The Economic Journal} 64(254), 290--323.

\bibitem {Phillips2}Phillips, A.~W. (1957). Stabilisation policy and the
time-forms of lagged responses. \textit{The Economic Journal} 67(266), 265--277.

\bibitem {Phillips 61}Phillips, A.~W. (1961). A simple model of employment,
money and prices in a growing economy. \textit{Economica} 28(112), 360--370.

\bibitem {Prescott}Prescott, E.~C. (1977). Should control theory be used for
economic stabilization? \textit{Carnegie-Rochester Conference Series on Public
Policy}, 7, 13--38.

\bibitem {Ramsey}Ramsey F.~P. (1927). A\ Contribution to the Theory of
Taxation. \textit{The Economic Journal} 37(145), 47--61.

\bibitem {Ramsey F2}Ramsey F.~P. (1928). A\ Mathematical Theory of Saving.
\textit{The Economic Journal} 38(152), 543--559.

\bibitem {Rogoff}Rogoff, K. (1985). The optimal degree of commitment to an
intermediate monetary target. \textit{The Quarterly Journal of Economics}
100(4), 1169--1189.

\bibitem {Rusnak}Rusn\'{a}k, M., Havranek, T., and Horv\'{a}th, R. (2013). How
to solve the price puzzle? A meta-analysis. \textit{Journal of Money, Credit
and Banking} 45(1), 37--70.

\bibitem {Sent}Sent, E.-M. (1998). Engineering dynamic economics.
\textit{History of Political Economy} 29 (Suppl), 41--62.

\bibitem {Sergi}Sergi, F. (2018). DSGE Models and the Lucas Critique. A
Historical Appraisal. UWE Bristol Economics Working Paper Series 1806.

\bibitem {Simaan Cruz}Simaan, M.~A. and Cruz, J.~B. (1973a). On the
Stackelberg Strategy in Non zero-Sum Games, \textit{Journal of Optimization
Theory and Applications} 11, 533--555.

\bibitem {Simaan Cruz 2}Simaan, M.~A. and Cruz, J.~B. (1973b). Additional
aspects of the Stackelberg strategy in nonzero-sum games. \textit{Journal of
Optimization Theory and Applications} 11, 613--626.

\bibitem {Simaan Cruz 3}Simaan M.~A. and Cruz J.~B. (1973c), A Stackelberg
Solution for Games with Many Players, \textit{IEEE Trans. on Automatic
Control} AC-18, 322--324.

\bibitem {Simon}Simon H.~A. (1956). Dynamic Programming under Uncertainty with
a Quadratic Criterion Function. \textit{Econometrica} 24(1), 74--81.

\bibitem {Simons}Simons, H.~C. (1936). Rules versus authorities in monetary
policy. \textit{Journal of Political Economy} 44(1), 1--30.

\bibitem {Sims}Sims, C.~A. (1980). Macroeconomics and reality.\textit{
Econometrica}, 1--48.

\bibitem {Singhal}Singhal, J. and Singhal, K. (2007). Holt, Modigliani, Muth,
and Simon's work and its role in the renaissance and evolution of operations
management. \textit{Journal of Operations Management} 25(2), 300--309.

\bibitem {Smith}Smith, A. (1776). \textit{The Wealth of Nations: An inquiry
into the nature and causes of the Wealth of Nations}. (first edition) W.
Strahan and T. Cadell, London.

\bibitem {Spear}Spear, S.~E. and Young, W. (2014). Optimum savings and optimal
growth: The Cass--Malinvaud--Koopmans nexus. \textit{Macroeconomic Dynamics}
18(1), 215--243.

\bibitem {Stengel}Stengel, R.~F. (1986) \textit{Optimal Control and
Estimation}. Dover publications, New York.

\bibitem {Tabellini}Tabellini, G. (2005). Finn Kydland and Edward Prescott's
contribution to the theory of macroeconomic policy. \textit{The Scandinavian
Journal of Economics} 107(2), 203--216.

\bibitem {Taylor 68}Taylor, J.~B. (1968). Fiscal and Monetary Stabilization
Policies in a Model of Endogenous Cyclical Growth. \textit{Princeton
Econometric Research Program Series}, October.

\bibitem {Taylor 83}Taylor J.~B. (1983). Comments: \textquotedblleft rules,
discretion and reputation in a model of monetary policy\textquotedblright,
\textit{Journal of Monetary economics} 12, 123--125.

\bibitem {Taylor 93}Taylor, J.~B. (1993). Discretion versus policy rules in
practice. \textit{Carnegie-Rochester conference series on public policy}. 39, 195--214.

\bibitem {Taylor 98}Taylor, J.~B. (1998). Applying academic research on
monetary policy rules: an exercise in translational economics. \textit{The
Manchester School} 66(S), 1--16.

\bibitem {Taylor}Taylor J.~B. (1999). The Robustness and Efficiency of
Monetary Policy Rules as Guidelines for Interest Rate Setting by the European
Central Bank. \textit{Journal of Monetary Economics.} 43, 655--679.

\bibitem {Taylor 2007}Taylor J.~B. (2007). Thirty-five Years of Model Building
for Monetary Policy Evaluation: Breakthroughs, Dark Ages and a Renaissance.
\textit{Journal of Money, Credit and Banking} 39(1), 193--201.

\bibitem {Taylor 2017}Taylor, J.~B. (2017). Rules versus discretion: assessing
the debate over the conduct of monetary policy (No. w24149). National Bureau
of Economic Research.

\bibitem {Theil}Theil, H. (1957). A note on certainty equivalence in dynamic
planning. \textit{Econometrica} 25(2), 346--349.

\bibitem {Tinbergen}Tinbergen J. (1952). \textit{On the theory of economic
policy. }North Holland, Amsterdam.

\bibitem {Turnovsky}Turnovsky S~.B. (2011). Stabilization Theory and Policy:
50 years after the Phillips Curve. \textit{Economica} 78, 67--88.

\bibitem {Tustin}Tustin A. (1953), \textit{The Mechanism of Economic Systems.
An approach to the problem of economic stabilisation from the point of view of
control system engineering}. London: Heinemann (1st edition), Cambridge, MA. :
Harvard Univ. Press, (2e ed. 1957).

\bibitem {von Zur Muehlen}Von Zur Muehlen, P. (1982). Activist vs.
non-activist monetary policy: optimal rules under extreme uncertainty. Working
paper. Federal Reserve Board, Washington.

\bibitem {Walsh}Walsh, C.~E. (2017). \textit{Monetary theory and policy}. MIT
press, Cambridge, Mass.

\bibitem {Whittle}Whittle, P. (1963). \textit{Prediction and regulation by
linear least-square methods.} English Univ. Press, London.

\bibitem {Wonham}Wonham, W.~M. (1974). \textit{Linear multivariable control. A
Geometric Approach} (1st edition). Springer, Berlin, Heidelberg.

\bibitem {Yellen}Yellen J. (2012). Perspective on Monetary Policy. Speech 6th
june, at the Boston Economic Club dinner, Boston, Massachussetts.

\bibitem {Sworder}Sworder, D. (1966). \textit{Optimal Adaptive Control Systems
by David Sworder}. Elsevier, Amsterdam.

\bibitem {Zhou}Zhou, K., Doyle, J.~C., and Glover, K. (1996). \textit{Robust
and optimal control} (Vol.~40, 146). Prentice hall, New
Jersey.\hypertarget{https://doi/10.1111/j.1468-0335.2009.00807.x}{}
\end{thebibliography}
\end{document}